\documentclass[12pt]{article}
\usepackage{makeidx}
\usepackage{amsfonts}
\usepackage{amsmath}
\usepackage{epsfig}
\usepackage{latexsym}
\usepackage{amssymb}
\usepackage{graphicx}

\newtheorem{theorem}{Theorem}

\newtheorem{lemma}[theorem]{Lemma}

\newtheorem{proposition}[theorem]{Proposition}
\newtheorem{remark}[theorem]{Remark}

\newenvironment{proof}[1][Proof]{\noindent\textbf{#1.} }{\ \rule{0.5em}{0.5em}}

\def\1{{\bf 1}}

\begin{document}

\title{Gaussian optimizers and the additivity problem in quantum information theory}
\author{A. S. Holevo \\
Steklov Mathematical Institute, Moscow}
\date{}
\maketitle

\begin{abstract}
We give a survey of the two remarkable analytical problems of quantum
information theory. The main part is a detailed report of the recent
(partial) solution of the quantum Gaussian optimizers problem which establishes an optimal property of Glauber's coherent
states -- a particular instance of pure quantum Gaussian states. We
elaborate on the notion of quantum Gaussian channel as a noncommutative
generalization of Gaussian kernel to show that the coherent states, and
under certain conditions only they, minimize a broad class of the concave
functionals of the output of a Gaussian channel. Thus, the output states
corresponding to the Gaussian input are ``the least chaotic'', majorizing
all the other outputs. The solution, however, is essentially restricted to
the gauge-invariant case where a distinguished complex structure plays a special role.

We also comment on the related famous additivity conjecture, which was
solved in principle in the negative some five years ago. This
refers to the additivity or multiplicativity (with respect to tensor products
of channels) of information quantities related to the classical capacity of
quantum channel, such as $(1\rightarrow p)$-norms or the minimal von Neumann
or R\'enyi output entropies. A remarkable corollary of the present
solution of the quantum Gaussian optimizers problem is that these additivity
properties, while not valid in general, do hold in the important and interesting class of the
gauge-covariant Gaussian channels.
\end{abstract}

\tableofcontents

\section{Introduction}

\label{sec:intro}

The quantum Gaussian optimizers problem is an analytical
problem that arose in quantum information theory at the end of past century, and
which has an independent mathematical interest. Only recently a
solution was found \cite{ghg}, \cite{mgh} in a considerably common
situation, while in full generality the problem still remains open. To
explain the nature and the difficulty of the problem we start from the
related classical problem of Gaussian maximizers which has been studied
rather exhaustively, see Lieb \cite{lieb} and references therein. Consider
an integral operator $G$ from $L_{p}\left( \mathbb{R}^{s}\right) $ to $
L_{q}\left( \mathbb{R}^{r}\right) $ given by a Gaussian kernels (i.e.
exponential of a quadratic form) with the $\left( q\rightarrow p\right) -$
norm
\begin{equation}
\left\Vert G\right\Vert _{q\rightarrow p}=\sup_{f\neq 0}\left\Vert
Gf\right\Vert _{p}/\left\Vert f\right\Vert _{q}=\sup_{\left\Vert
f\right\Vert _{q}\leq 1}\left\Vert Gf\right\Vert _{p}.  \label{clgaus}
\end{equation}
Under certain broad enough assumptions concerning the quadratic form
defining the kernel, and also $p$ and $q$, this operator is correctly
defined, and the supremum in (\ref{clgaus}) is attained on Gaussian $f$.
Moreover, under some additional restrictions any maximizer is Gaussian. As
it is put in the title of the paper \cite{lieb}: `` Gaussian kernels have
only Gaussian maximizers''.

Knowledge that the maximizer is Gaussian can be used to compute exact value
of the norm (\ref{clgaus}); in fact a starting point of the classical
Gaussian maximizers works were the result of K.I. Babenko \cite{bab} and a
subsequent paper of Beckner \cite{beck} which established the best constant
in the Hausdorff-Young inequality concerning the $\left( p\rightarrow
p^{\prime }\right) -$norm, ($p^{-1}+\left( p^{\prime }\right)
^{-1}=1,\,1<p\leq 2),$ of the Fourier transform (which is apparently given
by a degenerate imaginary Gaussian kernel).

A difficulty in the optimization problem (\ref{clgaus}) is that it requires
\textit{maximization} of a convex function, so the general theory of convex
optimization is not of great use here (it only implies that a maximizer of $
\left\Vert Gf\right\Vert _{p}$ belongs to a face of the convex set $
\left\Vert f\right\Vert _{q}\leq 1$). Instead, the solution is based on
substantial use of the classical Minkovski's inequality and the related
multipicativity of the classical $\left( q\rightarrow p\right) -$norms with
respect to tensor products of the integral operators.

A notable application of these classical results to a problem in quantum
mathematical physics was Lieb's solution \cite{lieb1} of Wehrl's conjecture
\cite{wehrl}. Let $\rho $ be a density operator in a separable Hilbert space
$ \mathcal{H}$ representing state of a quantum system; the ``classical
entropy'' of the state $\rho $ is defined as \footnote{
Throughout the paper the base of logarithm is a fixed number $a>1$. In
information theory the natural choice is $a=2$, then all the entropic
quantities are measured in ``bits''.}
\begin{equation*}
H_{cl}(\rho )=-\int_{\mathbb{C}}p_{\rho}(z) \log p_{\rho}(z) \frac{d^{2}z}{
\pi },
\end{equation*}
where $p_{\rho}(z)=\langle z |\rho |z\rangle $ is the diagonal value of the
kernel of $\rho $ in the system of Glauber's coherent vectors\footnote{
In analysis, they correspond to complex-parametrized Gaussian wavelets.
Notice that this is the only place in the present article where we formally
used Dirac's notations, uncommon among mathematicians.} $\{|z\rangle ; z\in
\mathbb{C}\}$ \cite{KS}, \cite{aspekty}.
The conjecture was that $H_{cl}(\rho )$ has the minimal value if $\rho $ is
itself a coherent state i.e. projector onto one of the coherent vectors.
Lieb \cite{lieb1} used exact constants in the Hausdorff-Young inequality for
$L_p$-norms of Fourier transform \cite{bab}, \cite{beck} and the Young
inequality for convolution \cite{beck} to prove similar maximizer conjecture
for $f(x)=x^{p}$ and considered the limit $\lim_{p\downarrow
1}(1-p)^{-1}\left( 1-x^{p}\right) =-\,x\log x$.

Recently, Lieb and Solovej \cite{ls}, by using a completely different
approach based on study of the spin coherent states, strengthened the result
of \cite{lieb1} by showing that the coherent states minimize \textit{any}
functional of the form $\int_{\mathbb{C}}f(p_{\rho}(z) )\frac{d^{2}z}{\pi },$
where $f(x),x\in [ 0,1]$ is a nonnegative concave function with $f(0)=0$.

In the language of quantum information theory, the affine map $G:\rho
\rightarrow p_{\rho}(z) ,$ taking density operators $\rho $ (quantum states)
into probability densities $p_{\rho}(z) $ (classical states), is a ``
quantum-classical channel'' \cite{h}. Moreover, it transforms Gaussian
density operators $\rho $ (in the sense defined below in Sec. \ref{sec:gaus}
) into Gaussian probability densities, and in this sense it is a `` Gaussian
channel'' . From this point of view, Wehrl entropy $H_{cl}(\rho )$ is the
output entropy of the channel, and Lieb's result says that it is minimized
by pure Gaussian states $\rho $. Moreover, the corresponding result for $
f(x)=x^{p}$ can be interpreted as `` Gaussian maximizer'' statement for the
norm $\left\Vert G\right\Vert _{1\rightarrow p}.$ Notice that the case $q=1,$
which is excluded in the classical problem for obvious reasons, appears and
is the most relevant in the quantum (noncommutative) case.

The quantum Gaussian optimizers problem described in the present paper
refers to Bosonic Gaussian channels -- a noncommutative analog of Gaussian
Markov kernels and, similarly, requires maximization of convex functions (or
minimization of concave functions, such as entropy) of the output state of
the channel, while the argument is the input state. A general conjecture is
that the optimizers belongs to the class of pure Gaussian states. The
conjecture, first formulated in \cite{HW} in the context of quantum
information theory, however natural it looks, resisted numerous attacks for
several years. Among others, notable achievements were the exact solution
for the classical capacity of pure loss channel \cite{Gio} and a proof of
additivity of the R\'enyi entropies of integer orders $p$ \cite{gl} for
special channels models. Even restricted to the class of Gaussian input
states, the optimization problem turns out to be nontrivial \cite{Ser}, \cite
{Hir}. There was some hope that in solving the problem, similarly to Wehrl's
conjecture, one could also use the classical `` Gaussian maximizers''
results. However the solution found recently by Giovannetti, Holevo,
Garcia-Patron \cite{ghg}, and Mari, Giovannetti, Holevo \cite{mgh} uses
completely different ideas based on a thorough study of structural
properties of quantum Gaussian channels. As it was mentioned, a solution of
the classical problem uses the Minkowski inequality and the implied
multiplicativity of $\left( q\rightarrow p\right) $-norms. However, the
noncommutative analog of the Minkowski's inequality \cite{carlieb} is not
powerful enough to guarantee the multipicativity of norms (or additivity of
the corresponding entropic quantities). Moreover, the related long-standing
additivity problem in quantum information theory \cite{Ho} was recently
shown to have negative solution in general \cite{hast}. We show that, remarkably, a
solution of the quantum Gaussian optimizers problem given in \cite{ghg}
implies also a proof of the multipicativity/additivity property in the
restricted class of gauge-covariant or contravariant quantum Gaussian
channels.

It would then be interesting to investigate a possible development of such
an approach to obtain noncommutative generalizations of the classical ``
Gaussian maximizers'' results for $\left(q\rightarrow p\right) -$norms. Such
generalization could shed a new light to the hypercontractivity problem for
quantum dynamical semigroups and related noncommutative analogs of
logarithmic Sobolev inequalities, see e.g. \cite{hyper}.

\section{The additivity problem for quantum channels}

\subsection{Definition of channel}

Let $\mathcal{H}$ be a separable complex Hilbert space, $\mathfrak{L}(
\mathcal{H} ) $ the algebra of all bounded operators in $\mathcal{H}$ and $
\mathfrak{T }(\mathcal{H})$ the ideal of trace-class operators. The space $
\mathfrak{T}( \mathcal{H})$ equipped with the trace norm $\left\Vert \cdot
\right\Vert _{1} $ is Banach space, which is useful to consider as a
noncommutative analog of the space $L_{1}.$ The convex subset of $\mathfrak{T
}(\mathcal{H})$
\begin{equation*}
\mathfrak{S}(\mathcal{H})=\left\{ \rho :\rho ^{\ast }=\rho \geq 0,\mathrm{Tr}
\rho =1\right\} ,
\end{equation*}
is a base of the positive cone in $\mathfrak{T}(\mathcal{H})$. Operators $
\rho $ from $\mathfrak{S}(\mathcal{H})$ are called \textit{density operators}
or \textit{quantum states}. The state space is a convex set with the extreme
boundary
\begin{equation*}
\mathfrak{P}(\mathcal{H})=\left\{ \rho :\rho \geq 0,\mathrm{Tr}\rho =1,\rho
^{2}=\rho \right\} .
\end{equation*}
Thus extreme points of $\mathfrak{S}(\mathcal{H}),$ which are called \textit{
pure states}, are one-dimensional projectors, $\rho =P_{\psi }$ for a vector
$\psi \in \mathcal{H}$ with unit norm, see, e.g. \cite{reed}.

The class of maps we will be interested is a noncommutative analog of Markov
maps (linear, positive, normalized maps) in classical analysis and
probability. Let $\mathcal{H}_{A},\mathcal{H}_{B}$ be the two Hilbert
spaces, which will be called input and output space, correspondingly. A map $
\Phi :\mathfrak{T}(\mathcal{H}_{A})\rightarrow \mathfrak{\ T}(\mathcal{H}
_{B})$ is \textit{positive} if $X\geq 0$ implies $\Phi [ X]\geq 0$, and it
is completely positive \cite{stine}, \cite{helems} if the maps $\Phi \otimes
\mathrm{Id}_{\left( d\right) }$ are positive for all $d=1,2,\dots ,$ where $
\mathrm{Id}_{\left( d\right) }$ is the identity map of the algebra $
\mathfrak{L}_{d}=\mathfrak{L}(\mathbb{C}^{d})$ of complex $d\times d-$
matrices. Equivalently, for every nonnegative definite block matrix $\left[
X_{jk}\right] _{j,k=1,\dots ,d}$ the matrix $\left[ \Phi [ X_{jk}]\right]
_{j,k=1,\dots ,d}$ is nonnegative definite.

A linear map $\Phi $ is trace-preserving if $\mathrm{Tr}\Phi [ X]=$ $\mathrm{
Tr} X$ for all $X\in\mathfrak{T }(\mathcal{H}_{A}).$

\textbf{Definition}
\textit{Quantum channel} is a linear completely positive trace-preserving
map $\Phi :\mathfrak{T}(\mathcal{H}_{A})\rightarrow \mathfrak{T}(\mathcal{H}
_{B}).$ Letter $A$ will be always associated with the \textit{input} of the
channel, while $B$ with the \textit{output}. Sometimes, to abbreviate
notations, we will write simply $\Phi :A\rightarrow B.$ $\blacksquare$

Apparently, every channel is a positive map taking states into states: $\Phi
[ \mathfrak{S}(\mathcal{H}_{A})]\subseteq\mathfrak{S}(\mathcal{H}_{B}).$
Since $\mathfrak{T}(\mathcal{H})$ is a base-normed space, this implies \cite
{davies} that $\Phi $ is a bounded map from the Banach space $\mathfrak{T}(
\mathcal{H}_{A})$ to $\mathfrak{T}(\mathcal{H}_{B}).$ The dual $\Phi ^{\ast }
$ of the map $\Phi $ is uniquely defined by the relation
\begin{equation}
\mathrm{Tr}\Phi [ X]Y=\mathrm{Tr}X\Phi ^{\ast }[Y];\quad X\in \mathfrak{T}(
\mathcal{H}_{A}),\,Y\in \mathfrak{L}(\mathcal{H}_{B}),  \label{dual}
\end{equation}
and it is called \textit{dual channel}. The dual channel is linear
completely positive $* -$weakly continuous map from $\mathfrak{L}(\mathcal{H}
_{B})$ to $\mathfrak{L}(\mathcal{H}_{A}),$ which is \textit{unital}: $\Phi [
I_{\mathcal{H}_{B}}]=I_{\mathcal{H}_{A}}.$ Here and in what follows $I$ with
possible index denotes the unit operator in the corresponding Hilbert space.

There are positive maps that are not completely positive, a basic example
provided by matrix transposition $X\rightarrow X^{\top }$ in a fixed basis.

From the definition of complete positivity one easily derives \cite{h} that
composition of channels $\Phi _{2}\circ \Phi _{1}$ defined as
\begin{equation*}
\Phi _{2}\circ \Phi _{1}[X]=\Phi _{2}[\Phi _{1}[X]],
\end{equation*}
and naturally defined tensor product of channels
\begin{equation*}
\Phi _{1}\otimes \Phi _{2}=\left( \Phi _{1}\otimes \mathrm{Id}_{2}\right)
\circ \left( \mathrm{Id}_{1}\otimes \Phi _{2}\right)
\end{equation*}
are again channels.\bigskip

\subsection{Stinespring-type representation}

The notion of completely positive map was introduced by Stinespring \cite
{stine} in a much wider context of C*-algebras. This allows also to cover
the notion of \textit{hybrid} channel where the input is quantum while the
output is classical or vice versa. An example of such channel was mentioned
in Sec. \ref{sec:intro}. We will not pursue this topic further here, see
\cite{h}, but only mention that complete positivity reduces to positivity in
such cases.

Motivated by the famous Naimark's dilation theorem, Stinespring established
a representation for completely positive maps of C*-algebras which in the
case of quantum channel reduces \cite{h} to

\begin{proposition}
Let $\Phi :A\rightarrow B $ be a channel. There exist a Hilbert space $
\mathcal{H}_{E} $ and an isometric operator $V:\mathcal{H}_{A}\rightarrow
\mathcal{H}_{B}\otimes \mathcal{H}_{E}$, such that
\begin{equation}
\Phi [ \rho ]=\mathrm{Tr}_{E}V\rho V^{\ast };\quad \rho \in \mathfrak{T}
\left( \mathcal{H}_{A}\right) ,  \label{stinespring}
\end{equation}
where $\mathrm{Tr}_{E}$ denotes partial trace with respect to $\mathcal{H}
_{E}.$ The representation (\ref{stinespring}) is not unique, however any two
representations with $V_{1}:\mathcal{H}_{A}\rightarrow \mathcal{H}
_{B}\otimes \mathcal{H}_{E_{1}}$ and $V_{2}:\mathcal{H}_{A}\rightarrow
\mathcal{H}_{B}\otimes \mathcal{H}_{E_{2}}$ are related via partial isometry
$W:\mathcal{H}_{E_{1}}\rightarrow \mathcal{H}_{E_{2}}$ such that $
V_{2}=\left( I_{B}\otimes W\right) V_{1}$ and $V_{1}=\left( I_{B}\otimes
W^{\ast }\right) V_{2}.$
\end{proposition}

Consider a representation (\ref{stinespring}) for the channel $\Phi ;$ the
\textit{complementary channel} \cite{H2}, \cite{KMNR} is then defined by the
relation
\begin{equation}
\tilde{\Phi}[\rho ]=\mathrm{Tr}_{B}V\rho V^{\ast };\quad \rho \in \mathfrak{T
}\left( \mathcal{H}_{A}\right) .  \label{complementary}
\end{equation}
From the relation between the different representations (\ref{stinespring}),
it follows that the complementary channel is unique in the following sense:
any two channels $\tilde{\Phi}_{1},\tilde{\Phi}_{2}$ complementary to $\Phi $
are isometrically equivalent in the sense that there is a partial isometry $
W:\mathcal{H}_{E_{1}}\rightarrow \mathcal{H}_{E_{2}}$ such that
\begin{equation}
\tilde{\Phi}_{2}[\rho ]=W\tilde{\Phi}_{1}[\rho ]W^{\ast },\quad \tilde{\Phi}
_{1}[\rho ]=W^{\ast }\tilde{\Phi}_{2}[\rho ]W,  \label{equiv}
\end{equation}
for all $\rho .$ It follows that the initial projector $W^{\ast }W$
satisfies $\tilde{\Phi}_{1}[\rho ]=W^{\ast }W\tilde{\Phi}_{1}[\rho ],$ i.e.
its support contains the support of $\tilde{\Phi}_{1}[\rho ],$ while the
final projector $WW^{\ast }$ has similar property with respect to $\tilde{
\Phi}_{2}[\rho ].$ The complementary to complementary can be shown
isometrically equivalent to the initial channel, so that $\Phi , \tilde{\Phi}
$ are called mutually complementary channels.

In general, we will say that two density operators $\rho $ and $\sigma $
(possibly acting in different Hilbert spaces) are \textit{isometrically
equivalent} if there is a partial isometry $W$ such that $\rho =W\sigma
W^{\ast },\quad \sigma =W^{\ast }\rho W.$ Apparently, this is the case if
and only if nonzero spectra (counting multiplicity) of the density operators
$\rho $ and $\sigma $ coincide. We denote this fact with the notation $\rho
\thicksim \sigma $. We have just shown that $\tilde{\Phi}_{1}[\rho
]\thicksim \tilde{\Phi}_{2}[\rho ]$ for arbitrary $\rho .$

\begin{lemma}
\label{L1} Let $\tilde{\Phi}$ be a complementary channel (\ref{complementary}
), then $\Phi [ P_{\psi }]\thicksim \tilde{\Phi}[P_{\psi }]$ for all $\psi
\in \mathcal{H}_{A}.$
\end{lemma}

\begin{proof}
Let $V:\mathcal{H}_{A}\rightarrow \mathcal{H}_{B}\otimes \mathcal{H}_{E}$ be
the isometry from the representations (\ref{stinespring}), (\ref
{complementary}), then $\rho _{BE}=VP_{\psi }V^{\ast }$ is a pure state in $
\mathcal{H}_{B}\otimes \mathcal{H}_{E},$ and the statement follows from a
basic result in quantum information theory (``Schmidt decomposition''): if $
\rho _{BE}$ is a pure state in $\mathcal{H}_{B}\otimes \mathcal{H}_{E}$ and $
\rho _{B}=\mathrm{Tr}_{E}\rho _{BE},\,\rho _{E}=\mathrm{Tr}_{B}\rho _{BE}$
are its partial states, then $\rho _{B}\thicksim \rho _{E}$ (see e.g.
Proposition 3 in \cite{Ho})
\end{proof}

A different name for channel is dynamical map -- in nonequilibrium quantum
statistical mechanics they arise as irreversible evolutions of an open
quantum system interacting with an environment \cite{h}. Assume that there
is a composite quantum system $AD=BE$ in the Hilbert space
\begin{equation}
\mathcal{H}=\mathcal{H}_{A}\otimes \mathcal{H}_{D}\simeq \mathcal{H}
_{B}\otimes \mathcal{H}_{E},
\end{equation}
which is initially prepared in the state $\rho _{A}\otimes \rho _{D}$ and
then evolves according to the unitary operator $U.$ Then the output state $
\rho_{B}$ depending on the input state $\rho _{A}=\rho$ is
\begin{equation}
\Phi _{B}[\rho ]=\mathop{\rm Tr}\nolimits_{E}U(\rho \otimes \rho
_{D})U^{\ast },  \label{fia}
\end{equation}
while the output state of the ``environment'' $E$ is the output of the
channel
\begin{equation}
\Phi _{E}[\rho ]=\mathop{\rm Tr}\nolimits_{B}U(\rho \otimes \rho
_{D})U^{\ast }.  \label{fie}
\end{equation}
If the initial state of $D$ is pure, $\rho _{D}=P_{\psi _{D}},$ then by
introducing the isometry $V: \mathcal{H}_{A}\rightarrow \mathcal{H}
_{B}\otimes \mathcal{H}_{D},$ which acts as
\begin{equation*}
V\psi =U(\psi \otimes \psi _{D} ),\quad \psi \in \mathcal{H}_{A},
\end{equation*}
we see that the relations (\ref{fia}), (\ref{fie}) convert into (\ref
{stinespring}), (\ref{complementary}), and $\Phi _{E}$ is just the
complementary of $\Phi _{B}.$ Notice also that both partial trace and
unitary evolution are completely positive operators, hence the maps (\ref
{fia}), (\ref{fie}) are completely positive; vice versa, any quantum channel
has a representation of such a form, see, e.g. \cite{h}.

Vast literature is devoted to study of quantum dynamical semigroups
(noncommutative analog of Markov semigroups) and quantum Markov processes.
Stinespring-type representation (\ref{stinespring}) underlies dilations of
quantum dynamical semigroups to the unitary dynamics of open quantum system
interacting with an environment \cite{davies}, \cite{structure}.

\subsection{Entropic quantities and additivity}

Consider the norm of the map $\Phi $ defined similarly to (\ref{clgaus}):
\begin{equation}
\left\Vert \Phi \right\Vert _{1\rightarrow p}=\sup_{X\neq 0}\left\Vert \Phi
[ X]\right\Vert _{p}/\left\Vert X\right\Vert _{1}=\sup_{\left\Vert
X\right\Vert _{1}\leq 1}\left\Vert \Phi [ X]\right\Vert _{p},  \label{norm}
\end{equation}
where $\left\Vert \cdot \right\Vert _{p}$ is the Schatten $p-$norm \cite
{reed}. As shown in \cite{aud},
\begin{equation}
\left\Vert \Phi \right\Vert _{1\rightarrow p}^{p}=\sup_{\rho \in \mathfrak{S}
\left( \mathcal{H}_{A}\right) }\mathrm{Tr}\Phi [ \rho ]^{p}=\sup_{\psi \in
\mathcal{H}_{A}}\mathrm{Tr}\Phi [P_{\psi} ]^{p},  \label{purity}
\end{equation}
where the second equality follows from convexity of the function  $x^p, p>1$..

The quantum R\'{e}nyi entropy of order $p>1$ of a density operator $\rho $
is defined as
\begin{equation}
R_{p}(\rho )=\frac{1}{1-p}\log \mathrm{Tr}\rho ^{p}=\frac{p}{1-p}\log
\left\Vert \rho \right\Vert _{p},  \label{qri}
\end{equation}
Define the \textit{minimal output R\'{e}nyi entropy} of the channel $\Phi $
\begin{equation}
\check{R}_{p}(\Phi )=\inf_{\rho \in \mathfrak{S}(\mathcal{H})}R_{p}(\Phi [
\rho ])=\frac{p}{1-p}\log \left\Vert \Phi \right\Vert _{1\rightarrow p}
\label{4p}
\end{equation}
and the \textit{minimal output von Neumann entropy}
\begin{equation}
\check{H}(\Phi )=\inf_{\rho \in \mathfrak{S}(\mathcal{H})}H(\Phi [ \rho ]).
\label{3}
\end{equation}
In the limit $p\rightarrow 1$ the quantum R\'{e}nyi entropies monotonely
nondecreasing converge to the von Neumann entropy
\begin{equation*}
\lim_{p\rightarrow 1}R_{p}(\rho )=-\mathrm{Tr}\rho \log \rho =H(\rho ).
\end{equation*}
In finite dimensions the set of quantum states is compact, hence by
Dini's Lemma the minimal output R\'{e}nyi entropies converge to the minimal
output von Neumann entropy\footnote{The corresponding statement is not valid for infinite-dimensional channels (even for classical channels with countable set of states), M. E. Shirokov, private communication.}.

Multiplicativity of the norm (\ref{norm}) for some channels $\Phi _{1},\Phi _{2}$,
\begin{equation}
\left\Vert \Phi _{1}\otimes \Phi _{2}\right\Vert _{1\rightarrow
p}=\left\Vert \Phi _{1}\right\Vert _{1\rightarrow p}\cdot \left\Vert \Phi
_{2}\right\Vert _{1\rightarrow p}  \label{multi}
\end{equation}
is equivalent to the additivity of the minimal output R\'{e}nyi entropies
\begin{equation}
\check{R}_{p}\left( \Phi _{1}\otimes \Phi _{2}\right) =\check{R}_{p}\left(
\Phi _{1}\right) +\check{R}_{p}\left( \Phi _{2}\right) .  \label{n2p}
\end{equation}
Closely related is the similar property for the minimal output von Neumann
entropy:
\begin{equation}
\check{H}(\Phi _{1}\otimes \Phi _{2})=\check{H}(\Phi _{1})+\check{H}(\Phi
_{2}).  \label{4}
\end{equation}
In finite dimensions, the validity of (\ref{n2p}) for certain channels
$\Phi _{1},\Phi _{2}$ and $p$ close to 1 implies (\ref{4}) for these channels.

In the last two relations the inequality $\leq $ (similarly to the
inequality $\geq $ in (\ref{multi})) is obvious because the right-hand side
is equal to the infimum over the subset of product states $
\rho=\rho_1\otimes\rho_2$. On the other hand, existence of ``entangled''
pure states which are not reducible to product states, is the cause for
possible violation of the equality for quantum channels.

\subsection{The channel capacity}

\label{sec:chcap}

The practical importance of the additivity property (\ref{4}) is revealed in
connection with the notion of the channel capacity. To explain it we assume
that $\mathcal{H}_{A},\mathcal{H}_{B}$ are finite dimensional for the moment.

For a quantum channel $\Phi $, a noncommutative analog of the Shannon
capacity, which we call $\chi -$\textit{capacity}, is defined by
\begin{equation}
C_{\chi }(\Phi )=\sup_{\left\{ \pi _{j},\rho _{j}\right\} }\left( H\left(
\Phi \left[ \sum_{j}\pi _{j}\rho _{j}\right] \right) -\sum_{j}\pi _{j}H(\Phi
[ \rho _{j}])\right) ,  \label{1}
\end{equation}
where the supremum is over all \textit{quantum ensembles}, that is finite
collections of states $\{\rho _{1},\ldots ,\rho _{n}\}$ with corresponding
probabilities $\{\pi _{1},\ldots ,\pi _{n}\}$. The quantity (\ref{1}) is
closely related to the capacity $C(\Phi )$ of quantum channel $\Phi $ for
transmitting classical information \cite{Ho}. The \textit{classical capacity}
of a quantum channel is defined as the maximal transmission rate per use of
the channel, with coding and decoding chosen for increasing number $n$ of
independent uses of the channel
\begin{equation*}
\Phi ^{\otimes n}=\underset{n}{\underbrace{\Phi \otimes \dots \otimes \Phi }}
\end{equation*}
such that the error probability goes to zero as $n\rightarrow \infty $ (for
a precise definition see \cite{h}). A basic result of quantum information
theory, HSW Theorem \cite{HOLEVO98}, says that such defined capacity $C(\Phi
)$ is related to $C_{\chi }(\Phi )$ by the formula
\begin{equation*}
C(\Phi )=\lim_{n\rightarrow \infty }(1/n)C_{\chi }(\Phi ^{\otimes n}).
\end{equation*}
Since $C_{\chi }(\Phi )$ is easily seen to be superadditive (i. e., $C_{\chi
}(\Phi _{1}\otimes \Phi _{2})\geq C_{\chi }(\Phi _{1})+C_{\chi }(\Phi _{2})$
), one has $C(\Phi )\geq C_{\chi }(\Phi )$. However if the additivity
\begin{equation}
C_{\chi }(\Phi _{1}\otimes \Phi _{2})=C_{\chi }(\Phi _{1})+C_{\chi }(\Phi
_{2})  \label{2}
\end{equation}
holds for a given channel $\Phi _{1}=\Phi $ and an arbitrary channel $\Phi
_{2}$, then
\begin{equation}
C_{\chi }(\Phi ^{\otimes n})=nC_{\chi }(\Phi ),  \label{wadd}
\end{equation}
implying
\begin{equation}
C(\Phi )=C_{\chi }(\Phi ).  \label{CCchi}
\end{equation}
The reason for possible violation of the equality here, as well as in the
cases (\ref{multi}), (\ref{n2p}), (\ref{4}), is existence of entangled
states, which are not reducible to product states, at the input of tensor
product channel $\Phi ^{\otimes n}$.

\subsection{Main conclusions}

Thus it was natural to ask: does the the additivity property (\ref{4}) holds
globally, i.e. for tensor product of \textit{any} pair of quantum channels $
\Phi _{1},\Phi _{2}$? The problem can be traced back to \cite{BFS}, see also
\cite{Ho}. Quite remarkably, Shor \cite{Shb}, see also \cite{FW}, had shown
the equivalence of the global properties of additivity of the $\chi -$
capacity and of the minimal output entropy.

\begin{theorem}
\label{glob} \cite{Shb} The properties (\ref{2}) and (\ref{4}) are globally
equivalent in the sense that if one of them holds for all channels $\Phi
_{1},\Phi _{2},$ then another is also true for all channels.
\end{theorem}

The additivity is proved rather simply for all classical channels (see e.g.
\cite{COVER}), but in the quantum case the question remained open for a
dozen of years, and was ultimately solved in the negative.

The detailed history of the problem up to 2006 can be found in \cite{Ho},
and here we only sketch the basic steps and the final resolution. In \cite
{ahw} it was suggested to approach the additivity property (\ref{4}) via
multiplicativity (\ref{multi}) of the $\left( 1\rightarrow p\right) -$norms
(equivalent to additivity (\ref{n2p}) of the minimal output R\'{e}nyi
entropies). The first explicit example where this property breaks for $d=
\mathrm{dim}\mathcal{H}\geq 3$ and large enough $p$ was \textit{
transpose-depolarizing channel} \cite{WH}:
\begin{equation}\label{ctd}
\Phi (\rho )=\frac{1}{d-1}\left[ I\,\mathrm{Tr}\rho -\rho ^{\top }\right] ,
\end{equation}
where $\rho\in \mathfrak{L}_d$ is a matrix and $\rho ^{\top }$ its
transpose. In particular, (\ref{n2p}) with $\Phi _{1}=\Phi _{2}=\Phi $ fails
to hold for $p\geq 4,7823$ if $d=3$ (nevertheless, the additivity of $\check{
H}(\Phi )$ and of $C_{\chi }(\Phi )$ holds for this channel). Five years
later came important findings of Winter \cite{wint} and Hayden \cite{hay},
see also \cite{HaWi}, who showed existence of a pair of channels breaking
the additivity of the minimal output R\'{e}nyi entropy for all values of the
parameter $p>1$. The method of these and subsequent works is random choice
of the channels, which for fixed dimensions are parametrized by isometries $V
$ in the representation (\ref{stinespring}), as well of the input states of
the channels, combined with sufficiently precise probabilistic estimates for
the norms (\ref{purity}). For finite dimensions the corresponding parametric
sets are compact, and one usually takes the uniform distribution. Basing on
this progress, Hastings \cite{hast} gave a proof of existence of channels
breaking the additivity conjecture (\ref{4}) corresponding to $p=1$, in very
high dimensions. Moreover, the probability of violation of the additivity
tends to 1 as the dimensionalities tend to infinity. Hastings gave only a
sketch, and the detailed proof following his approach was given by Fukuda,
King and Moser \cite{Fuk}, and further simplified by Brandao and M. Hordecki
\cite{brandao}. Later Szarek et al. \cite{sz} proposed a proof related to
the \textit{Dvoretzky-Mil'man theorem} on almost Euclidean sections of
high-dimensional convex bodies.

Although, combined with theorem \ref{glob} this gives a definite negative
answer to the additivity conjectures, several important issues remaine open.
All the proofs use the technique of random unitary channels or random states
and as such are not constructive: they prove only existence of
counterexamples but do not allow to actually produce them. Attempts to give
estimates for the dimensions in which nonadditivity can happen based on
Hastings' approach has led to overwhelmingly high values: the detailed
estimates made in \cite{Fuk} gave $ d\approx 10^{32}$
breaking the additivity by a quantity of the order $10^{-5}$. The best
result in this direction obtained in \cite{belin} states that ``violations
of the additivity of the minimal output entropy, using random unitary
channels and a maximally entangled state state, can occur if and only if the
output space has dimension at least 183. Almost surely, the defect of
additivity is less than $\log 2$, and it can be made as close as desired to $
\log 2$''.

While this does not exclude possibility of better estimates, based perhaps
on a different (but yet unknown) models, it casts doubt onto finding
concrete counterexamples by computer simulation of random channels.
From this point of view, the following explicit example given in \cite{Gru} is of interest.
Consider the completely positive map
\begin{equation*}
\rho\longrightarrow\Phi_{\_}[\rho ]=
\mathop{\rm Tr}\nolimits _{2}P_{-}\rho P_{-},\quad \rho \in\mathfrak{T}(\mathcal{H} \otimes \mathcal{H}),
\end{equation*}
where $P_{-}$ is the projector onto the antisymmetric
subspace $\mathcal{H} _{\_}$ of $\mathcal{H} \otimes \mathcal{H} $ which has
the dimensionality $\frac{d(d-1)}{2}$, and the partial trace is taken with respect to the second copy of $\mathcal{H}$.
Its restriction to the operators with support in the subspace $\mathcal{H}
_{\_}$ is trace preserving, hence it is a channel. It can be shown \cite{h} that $\Phi_{\_}=\frac{(d-1)}{2}\tilde{\Phi}^{\ast
} $ where $\tilde{\Phi}^{\ast }$ is the dual to the complementary of the channel (\ref{ctd}).
For this simple channel the minimal R\'{e}nyi
entropies are nonadditive for all $p>2$ and sufficiently large $d$, but unfortunately it is not clear if it could be extended
to the most interesting range $p\geq 1.$

Coming back to arbitrary channels, it remains unclear what happens in small dimensions: perhaps the
additivity still holds generically for some unknown reason, or its violation
is so tiny that it cannot be revealed by numerical simulations. This is
indeed surprising in view of the fact that the physical reason for
nonadditivity is entanglement between the inputs of the parallel quantum
channels, see \cite{h} for more detail.

On the other hand, these results stress the importance of continuing efforts
to find special cases where the additivity holds for some reason, and can be
established analytically.

A survey of the main classes of such `` additive'' channels acting in finite
dimensions was presented in \cite{Ho}; below we briefly list the most
important classes of channels $\Phi $ for which the additivity properties (
\ref{4}), (\ref{2}) and (\ref{n2p}) for $p>1$ were established with $\Phi
=\Phi_{1} $ and arbitrary $\Phi _{2}.$

\begin{itemize}
\item Qubit unital channels, i.e channels $\Phi :$ $\mathfrak{L}
_{2}\rightarrow \mathfrak{L}_{2}$ satisfying $\Phi [ I]=I$ \cite{Kib}.
Strikingly, there is still no analytical proof of the additivity for
nonunital qubit channels, in spite of a convincing numerical evidence \cite
{HIMRT}.

\item Depolarizing channel in $\mathfrak{L}_{d}:$
\begin{equation*}
\Phi [ \rho ]=\left( 1-p\right) \rho +p\frac{I}{d}\mathrm{Tr}\rho ,\quad
0\leq p\leq \frac{d^{2}}{d^{2}-1},
\end{equation*}
which is the only unitarily-covariant channel, and can be regarded as
noncommutative analog of completely symmetric channel in classical
information theory \cite{COVER}. The additivity properties (\ref{4}), (\ref
{n2p}), (\ref{2}) were proved by King \cite{Kic}.

\item Entanglement-breaking channels. In finite dimensions these are
channels of the form
\begin{equation*}
\Phi [ \rho ]=\sum\limits_{j}\rho _{B}\,\mathrm{Tr\,}\rho M_{A},
\end{equation*}
where $\left\{ M_{A}\right\} $ is a resolution of the identity in $\mathcal{H
}_{A}$: $M_{A}\geq 0,\,\sum\limits_{j}\,M_{A}=I_{A},$ and $\rho _{B}\in
\mathfrak{S}(\mathcal{H}_{B})$ (see \cite{R}).

For the finite-dimensional entanglement-breaking channels the additivity of
the minimal output von Neumann entropy and of the $\chi -$capacity was
established by Shor \cite{Sha} and the additivity of the minimal output R
\'{e}nyi entropies -- by King \cite{Keb}. The additivity properties of
entanglement-breaking channels were generalized to infinite dimensions by
Shirokov \cite{Sh-2}.

\item Complementary channels.

The additivity of the minimal output entropy is equivalent for a channel $
\Phi $ and its complementary $\tilde{\Phi}$, see Lemma \ref{L1a} below. The
class of channels complementary to entanglement-breaking contains the
Schur-multiplication maps of matrices $\rho =[c_{jk}]$ $_{j,k=1,\dots ,d}$
in $\mathfrak{L}_{d}$:
\begin{equation*}
\tilde{\Phi}[\rho ]=[\gamma _{jk}c_{jk}]_{j,k=1,\dots ,d},
\end{equation*}
where $[\gamma _{jk}]$ $_{j,k=1,\dots ,d}$ is a nonnegative definite matrix
such that $\gamma _{jj}\equiv 1.$ For these channels, which are also called
``Hadamard channels'' the additivity of the $\chi -$capacity was also
established \cite{KMNR}.
\end{itemize}

In the next Sections we consider Bosonic Gaussian channels which act in
infinite-dimensional spaces. One of the main goals of the present paper is to show
that the additivity holds for a wide class of gauge co- or contravariant Gaussian
channels, i.e. those which respect a fixed complex structure in the
underlying symplectic space.

\subsection{Majorization for quantum states}

From now on we again allow the Hilbert spaces in question to be
infinite-dimensional. Denote by $\mathfrak{F}$ the class of real concave
functions $f$ on $[0,1],$ such that $f(0)=0.$ For any $f\in \mathfrak{F}$
and for any density operator $\rho $ we can consider the quantity
\begin{equation*}
\mathrm{Tr}f(\rho )=\sum_{j}f(\lambda _{j}),
\end{equation*}
where $\lambda _{j}$ are the (nonzero) eigenvalues of the density operator $
\rho ,$ counting multiplicity. Note that this quantity is defined
unambiguously with values in $(-\infty ,\infty ].$ This follows from the
fact that $f(x)\geq cx,$ where $c=f(1),$ hence $\mathrm{Tr}f(\rho )\geq c
\mathrm{Tr}\rho =c.$ We also will use the fact that the functional $\rho
\rightarrow \mathrm{Tr}f(\rho )$ is (strictly) concave on $\mathfrak{S}(
\mathcal{H})$ if $f$ is (strictly) concave (see e.g. \cite{carlen}).

Denote by $\lambda _{j}^{\downarrow }(\rho )$ the eigenvalues of a density
operator $\rho ,$counting multiplicity, arranged in the nonincreasing order.
One says that density operator $\rho $ \textit{majorizes} density operator $
\sigma $ if
\begin{equation*}
\sum_{j=1}^{k}\lambda _{j}^{\downarrow }(\rho )\geq \sum_{j=1}^{k}\lambda
_{j}^{\downarrow }(\sigma ),\quad k=1,2,\dots
\end{equation*}
A consequence of a well known result, see e.g. \cite{carlen}, is that this
is the case if and only if $\mathrm{Tr}f(\rho )\leq \mathrm{Tr}f(\sigma )$
for all $f\in \mathfrak{F}\,.$\bigskip

For a quantum channel $\Phi $ we introduce the quantity
\begin{equation}
\check{f}(\Phi )=\inf_{\rho \in \mathfrak{S}(\mathcal{H})}\mathrm{Tr}f(\Phi
[ \rho ])=\inf_{P_{\psi }\in \mathfrak{P}(\mathcal{H})}\mathrm{Tr}f(\Phi [
P_{\psi }]),  \label{f_check}
\end{equation}
where the second equality follows from the concavity of the functional $\rho
\rightarrow \mathrm{Tr}f(\Phi [ \rho ])$ on $\mathfrak{S}(\mathcal{H}). $
Moreover, for strictly concave $f,$ any minimizer is of the form $P_{\psi } $
for some vector $\psi \in \mathcal{H}.$

In particular, taking $f(x)=-x\log x$ and $f(x)=-x^{p},$ we obtain $\check{f}
(\Phi )=\check{H}(\Phi )$ and $\check{f}(\Phi )=-\left\Vert \Phi \right\Vert
_{1\rightarrow p}^{p}.$

\begin{lemma}
\label{L1a} For complementary channels, $\check{f}(\Phi )=\check{f}(\tilde{
\Phi}).$ Hence $\left\Vert \Phi \right\Vert _{1\rightarrow p}=\Vert \tilde{
\Phi} \Vert _{1\rightarrow p}$, $\check{H}(\Phi )=\check{H}(\tilde{\Phi}),$ $
\check{R}_{p}(\Phi )=\check{R}_{p}(\tilde{\Phi}),$ and the multiplicativity (
\ref{multi}), as well as the additivity of the minimal output entropies (\ref
{4}), (\ref{n2p}) holds simultaneously for pairs of channels $\Phi _{1},
\Phi _{2}$  and $\tilde{\Phi}_{1}, \tilde{\Phi}_{2}$.
\end{lemma}

\begin{proof}
From Lemma \ref{L1}, $\Phi [ P_{\psi }]$ and $\tilde{\Phi}[P_{\psi }]$ have
identical nonzero spectrum ($\Phi [ P_{\psi }]\thicksim \tilde{\Phi}[P_{\psi
}]$) . Then
\begin{equation}
\mathrm{Tr}f(\Phi [ P_{\psi }])=\mathrm{Tr}f(\tilde{\Phi}[P_{\psi }])
\label{eq_comp}
\end{equation}
since $f(0)=0.$ Using second equality in (\ref{f_check}) implies $\check{f}
(\Phi )=\check{f}(\tilde{\Phi}).$

The statement about multiplicativity (additivity) then follows from the fact
that the channel $\tilde{\Phi}_{1}\otimes\tilde{\Phi}_{2}$ is complementary
to $\Phi _{1}\otimes\Phi _{2}$.
\end{proof}

\section{Quantum Gaussian systems}

\subsection{Gaussian states and channels}

\label{sec:gaus}

\bigskip A real vector space $Z$ equipped with a nondegenerate
skew-symmetric form $\Delta (z,z^{\prime })$ is called \textit{symplectic
space}. In what follows $Z$ is finite-dimensional, in which case its
dimensionality is necessarily even, \textrm{dim}$Z=2s$ \cite{kostr}. A basis
$\left\{ e_{j},h_{j};j=1,\dots ,s\right\} $ in which the form $\Delta
(z,z^{\prime })$ has the matrix
\begin{equation}
\Delta =\mathrm{diag}\left[
\begin{array}{cc}
0 & 1 \\
-1 & 0
\end{array}
\right] _{j=1,\dots ,s}  \label{delta}
\end{equation}
is called \textit{symplectic}. The \textit{Weyl system} in a Hilbert space $
\mathcal{H}$ is a strongly continuous family $\{ W(z); z\in Z\}$ of unitary
operators satisfying the \textit{\ Weyl-Segal canonical commutation relation
(CCR)}
\begin{equation}
W(z)W(z^{\prime })=\exp [-\frac{i}{2}\Delta (z,z^{\prime })]W(z+z^{\prime }).
\label{weyl}
\end{equation}
Thus $z\rightarrow W(z)$ is a projective representation of the additive
group of $Z$. We always assume that the representation is irreducible. The
Stone-von Neumann uniqueness theorem says that such a representation is
unique up to unitary equivalence. It is well-known, see e.g. \cite{reed},
that there is a family of selfadjoint operators \bigskip $z\rightarrow R(z)$
with a common essential domain $\mathcal{D}$ such that
\begin{equation*}
W(z)=\exp i\,R(z),
\end{equation*}
moreover, for any symplectic basis $\left\{ e_{j},h_{j};j=1,\dots ,s\right\}
$
\begin{equation*}
R(z)=\sum_{j=1}^{s}(x_{j}q_{j}+y_{j}p_{j})
\end{equation*}
on $\mathcal{D},$where $R(e_{j})=q_{j},\,R(h_{j})=p_{j},$ and $
[x_{1},y_{1},\dots ,x_{s},y_{s}]$ are coordinates of vector $z$ in the
basis. Here the \textit{canonical observables} $q_{j},p_{j};j=1,\dots ,s$
are selfadjoint operators in $\mathcal{H}$ satisfying the Heisenberg CCR on $
\mathcal{D}$
\begin{equation}
[ q_{j},p_{k}]\subseteq i\delta
_{jk}I,\;\;[q_{j},q_{k}]=0,\;\;[p_{j},p_{k}]=0.  \label{Hei}
\end{equation}
In physics the symplectic space is the phase space of the classical system
(such as electro-magnetic radiation modes in the cavity), the quantum
version of which is described by CCR. Then $s$ is number of degrees of
freedom, or ``normal modes'' of the classical system.

The state given by density operator $\rho $ in $\mathcal{H}$ is called
\textit{Gaussian}, if its \textit{quantum characteristic function}
\begin{equation*}
\phi (z)=\mathrm{Tr}\rho W(z)
\end{equation*}
has the form
\begin{equation}
\phi (z)=\exp \left( i\,m(z)-\frac{1}{2}\alpha \left( z,z\right)
\right) ,  \label{GaussianState}
\end{equation}
where $m$ is a real linear form and $\alpha $\ is a real bilinear symmetric
form on $Z $. A necessary and sufficient condition for (\ref{GaussianState})$
\ $to define a state is nonnegative definiteness of the (complex) Hermitian
form\footnote{
A complex-valued real-bilinear form $\beta \left( z,z^{\prime }\right) $ on $
Z$ will be called Hermitian if $\beta \left( z^{\prime },z\right) =\overline{
\beta \left( z,z^{\prime }\right) }.$} $\alpha \left( z,z^{\prime }\right) -
\frac{i}{2}\Delta \left( z,z^{\prime }\right) $ on $Z$ or, briefly:
\begin{equation}
\alpha \geq \frac{i}{2}\Delta .  \label{n-s condition}
\end{equation}
We will agree that the matrix of a bilinear form in fixed a symplectic base
is denoted by the same letter, then (\ref{n-s condition}) can be understood
as inequality for Hermitian matrices, where $\alpha $ is real symmetric and $
\Delta $ is real skew-symmetric.

A Gaussian state is pure if and only if $\alpha $ is a minimal solution of
this inequality, see e.g. \cite{hextr}. Operator $J$ in $Z$ is called
\textit{operator of complex structure} if
\begin{equation}
J^{2}=-I,  \label{j2e}
\end{equation}
where $I$ is the identity operator in $Z,$ and the bilinear form $\Delta
(z,Jz^{\prime })$ is an (Euclidean) inner product in $Z,$ i.e.
\begin{eqnarray}
\Delta (z,Jz^{\prime }) &=&\Delta (z^{\prime },Jz)\,(=-\Delta (Jz,z^{\prime
}));\quad  \label{j3e} \\
\Delta (z,Jz) &\geq &0,\quad z\in Z.
\end{eqnarray}
The following characterization can be found in \cite{ccr}, \cite{h}:

\begin{proposition}
The minimal solutions of the inequality (\ref{n-s condition}) are in
one-to-one correspondence with the operators $J$ of complex structure in $Z$
given by the relation
\begin{equation*}
\alpha \left( z,z^{\prime }\right) =\frac{1}{2}\Delta (z,Jz^{\prime });\quad
z,z^{\prime }\in Z.
\end{equation*}
\end{proposition}

In this way to every complex structure corresponds the family of pure
Gaussian states (\ref{GaussianState}) with different values of $m$ which are
called the $J$-\textit{coherent states.} The state with $m=0$ is called $J$-
\textit{vacuum}. Let $\rho _{0}$ be a vacuum, then any associated coherent
state is of the form $W(z^{\prime })\rho _{0}W(z^{\prime })^{\ast },$ as
follows from the relation
\begin{equation*}
W(z^{\prime })^{\ast }W(z)W(z^{\prime })=\exp [-i\Delta (z,z^{\prime })]W(z)
\end{equation*}
and from nondegeneracy of the form $\Delta (z,z^{\prime })$ due to which $
m(z)=\Delta (z,z_m^{\prime })$.

Operator $S$ in $Z$ is called \textit{symplectic} if $\Delta (Sz,Sz^{\prime
})=\Delta (z,z^{\prime })$ for all $z,z^{\prime }\in Z.$ The unitary
operators $W(Sz)$ satisfy the CCR (\ref{weyl}) hence by the Stone-von
Neumann uniqueness theorem there is a unitary operator $U_{S}$ in $\mathcal{H
}$\ such that
\begin{equation*}
W(Sz)=U_{S}^{\ast }W(z)U_{S},\quad z\in Z.
\end{equation*}
The map $S\rightarrow U_{S}$ is a projective representation of the group of
all symplectic transformations in $Z,$ sometimes called `` metaplectic
representation'' \cite{simon} as it can be extended to a faithful unitary
representation of the metaplectic group which is two-fold covering of the
symplectic group.

Similarly, $T$ is \textit{antisymplectic} if $\Delta (Tz,Tz^{\prime
})=-\Delta (z,z^{\prime })$ for all $z,z^{\prime }\in Z.$ There is an
antiunitary operator $U_{T}$ in $\mathcal{H}$ such that
\begin{equation*}
W(Tz)=U_{T}^{\ast }W(z)U_{T},\quad z\in Z.
\end{equation*}

Let $Z_{A},Z_{B}$ be two symplectic spaces with the corresponding Weyl
systems. Consider a channel $\Phi :A\longrightarrow B$. The channel is
called \textit{Gaussian} if the dual channel satisfies
\begin{equation}
\Phi ^{\ast }[W_{B}(z)]=W_{A}(Kz)\exp \left[ il(z)-\frac{1}{2}\mu (z,z)
\right] ,\quad z\in Z_{B},  \label{bosgaus}
\end{equation}
where $K:Z_{B}\rightarrow Z_{A}$ is a linear operator, $l$ a linear form and
$\mu $\ is a real symmetric form on $Z_{B}.$ In terms of characteristic
functions of states,
\begin{equation*}
\phi _{B}(z)=\phi _{A}(Kz)\exp \left[ il(z)-\frac{1}{2}\mu (z,z)\right] .
\end{equation*}
It follows that Gaussian channel maps Gaussian states into Gaussian states.
A converse statement also holds true \cite{ccr}.

A necessary and sufficient condition on parameters $(K,l,\mu )$ for complete
positivity of the map $\Phi $ is (see \cite{cgeh}) nonnegative definiteness
of the Hermitian form
\begin{equation*}
z,z^{\prime }\longrightarrow \mu \left( z,z^{\prime }\right) -\frac{i}{2}
\left[ \Delta _{B}\left( z,z^{\prime }\right) -\Delta _{A}\left(
Kz,Kz^{\prime }\right) \right]
\end{equation*}
on $Z_{B},$ or, in matrix terms (if some bases are chosen in $Z_{A},Z_{B}$),
\begin{equation}
\mu \geq \frac{i}{2}\left[ \Delta _{B}-K^{t}\Delta _{A}K\right] ,
\label{n-s channel}
\end{equation}
where $^{t}$ denotes transposition of a matrix. The proof using explicit
construction of the representation of type (\ref{fia}) is given in \cite
{cgeh}, see also \cite{h}; below in Proposition \ref{prop9} below we give
such a construction for an important particular class of Gaussian channels.

We call the Gaussian channel extreme\footnote{
In quantum optics one speaks of \textit{quantum-limited} channels \cite{gp}.}
if $\mu $ is a minimal solution of the inequality (\ref{n-s channel}). This
terminology stems from the fact that the minimality of $\mu $ is necessary
and sufficient for the channel $\Phi $ to be an extreme point in the convex
set of all channels with fixed input and output spaces \cite{hextr}.

\textbf{Additivity hypothesis for quantum Gaussian channels}: The additivity
properties (\ref{n2p}), (\ref{4}) hold for any pair of Gaussian channels $
\Phi _{1},\Phi _{2}.$

\textbf{Hypothesis of quantum Gaussian minimizers:} For any function $f\in
\mathfrak{F}$ the infimum in (\ref{f_check}) is attained on a pure Gaussian
state $\rho .$

Any Gaussian channel has the covariance property
\begin{equation}
\Phi [ W_{A}(z)\rho W_{A}(z)^{\ast }]=W_{B}(K^{s}z)\Phi [ \rho
]W_{B}(K^{s}z)^{\ast }  \label{gauscov}
\end{equation}
where $K^{s}$ is the symplectic adjoint operator defined by the relation
\begin{equation*}
\Delta _{B}\left( K^{s}z_{A},z_{B}\right) =\Delta _{A}\left(
z_{A},Kz_{B}\right) .
\end{equation*}
It follows that the value $\mathrm{Tr}f(\Phi [ \rho ])$ is the same for all
coherent states $W(z)\rho _{0}W(z)^{\ast }$ associated with a vacuum state $
\rho _{0}.$

These two problems turn out to be closely related. In what follows we
describe positive solution for both of them in a particular and important
class of Gaussian channels with gauge symmetry. However both conjectures
remain open for general quantum Gaussian channels.

\subsection{Complex structures and gauge symmetry}

\label{gaugau}

Given an operator of the complex structure $J$ one defines in $Z$ the
Euclidean inner product $j(z,z^{\prime })=\Delta (z,Jz^{\prime }).$ Then one
can define in $Z$ the structure of $s-$dimensional unitary space $\mathbf{Z}$
in which $i\mathbf{z}$ corresponds to $Jz$ and the (Hermitian) inner product\footnote{In accordance with convention accepted in mathematical physics, the inner product is
complex linear with respect to $\mathbf{z}^{\prime }$ and anti-linear with respect to $\mathbf{z}$.}
is
\begin{equation*}
\mathbf{j}(\mathbf{z},\mathbf{z}^{\prime })=\frac{1}{2}[\Delta (z,Jz^{\prime })+i\Delta
(z,z^{\prime })]=\frac{1}{2}[j(z,z^{\prime })-ij(z,Jz^{\prime })].
\end{equation*}

From (\ref{j2e}), (\ref{j3e}) it follows that $J$ is symplectic, that is $
\Delta (Jz,Jz^{\prime })=\Delta (z,z^{\prime })$ for all $z,z^{\prime }\in Z$
. With every complex structure one can associate the cyclic one-parameter
group of symplectic transformations $\left\{ \mathrm{e}^{\varphi J};\varphi
\in [ 0,2\pi )\right\} $  which we call the \textit{gauge group}. Hence, by
the Stone-von Neumann uniqueness theorem, the gauge group in $Z$ induces the
one-parameter unitary group of the \textit{gauge transformations} $\left\{
U_{\varphi };\varphi \in [ 0,2\pi )\right\} $ in $\mathcal{H}$ according to
the formula
\begin{equation}
W(\mathrm{e}^{\varphi J}z)=U_{\varphi }^{\ast }W(z)U_{\varphi }.
\label{gaugetr}
\end{equation}
For the future use it will be convenient to introduce the complex
parametrization of the Weyl operators by defining the \textit{displacement
operators}
\begin{equation}
D(\mathbf{z})=W(Jz),\quad \mathbf{z\in Z.}  \label{disp}
\end{equation}

A state $\rho $ is gauge invariant if $\rho =U_{\varphi }\rho U_{\varphi
}^{\ast }$ for all $\varphi $, which is equivalent to the property $\phi
(z)=\phi (\mathrm{e}^{\varphi J}z)$ of the characteristic function. In
particular, Gaussian state (\ref{GaussianState}) is gauge invariant if $
m(z)\equiv 0$ and $\alpha \left( z,z^{\prime }\right) =\alpha \left(
Jz,Jz^{\prime }\right) .$ By introducing the Hermitian inner product in $
\mathbf{Z}$
\begin{equation*}
\boldsymbol{\alpha }\left( \mathbf{z},\mathbf{z}^{\prime }\right) =\frac{1}{2}[\alpha
\left( z,z^{\prime }\right) - i\alpha \left( z,Jz^{\prime }\right)],
\end{equation*}
we have $\boldsymbol{\alpha }\left( \mathbf{z},\mathbf{z}\right) =\frac{1}{2}\alpha
\left( z,z\right) $ since $\alpha \left( z,Jz^{\prime }\right) $ is
skew-symmetric; moreover, the condition (\ref{n-s condition}) is equivalent
to nonnegative definiteness of the Hermitian form $\boldsymbol{\alpha }
\left( \mathbf{z},\mathbf{z}^{\prime }\right) -\frac{1}{2}\mathbf{j}(\mathbf{
z},\mathbf{z}^{\prime })$ on $\mathbf{Z}:$
\begin{equation}
\boldsymbol{\alpha }\geq \frac{1}{2}\mathbf{j}.  \label{n-s_comp}
\end{equation}
This follows from application of the following Lemma to the form
\begin{equation*}
z,z^{\prime }\longrightarrow \beta (z,z^{\prime })=\alpha (z,z^{\prime })-
\frac{i}{2}\Delta (z,z^{\prime }).
\end{equation*}
The relation (\ref{n-s_comp}) can be considered as the inequality for the
matrices of the form, provided a basis is chosen in $\mathbf{Z}$. In an
orthonormal basis, $\mathbf{j}=\mathbf{I}$ is the unit matrix.

\begin{lemma}
\label{L0} Let $\beta (z,z^{\prime })$ be a bilinear complex-valued
Hermitian form on real vector space $Z,$ satisfying $\beta (Jz,Jz^{\prime
})=\beta (z,z^{\prime }),$ where $J$ is a linear operator such that $
J^{2}=-I.$ Then $\beta (z,z^{\prime })$ is nonnegative definite i.e.
\begin{equation}
\sum_{jk}\bar{c}_{j}c_{k}\beta \left( z_{j},z_{k}\right) \geq 0  \label{U1}
\end{equation}
for any finite collection $\left\{ z_{j}\right\} \subset Z$ and any $\left\{
c_{j}\right\} \subset \mathbb{C}$, if and only if
\begin{equation}
\mathrm{Re}\beta (z,z)\pm \mathrm{Im}\beta (z,Jz)\geq 0\quad \text{for all \
}z\in Z.  \label{U2}
\end{equation}
\end{lemma}

\begin{proof}
(\ref{U2})$\Longrightarrow $(\ref{U1}): We have $\beta (z,z^{\prime })=
\mathrm{Re}\beta (z,z^{\prime })+i\mathrm{Im}\beta (z,z^{\prime }),$ where $
\mathrm{Im}\beta (z,z^{\prime })$ is skew-symmetric, hence $\mathrm{Im}\beta
(z,z)=0.$ By using the fact that $\beta (Jz,z^{\prime })=-\beta
(z,Jz^{\prime })$ we obtain that also $\mathrm{Re}\beta (z,Jz^{\prime })$ is
skew-symmetric, hence $\mathrm{Re}\beta (z,Jz)=0.$ Thus
\begin{equation*}
\mathrm{Re}\beta (z,z)\pm \mathrm{Im}\beta (z,Jz)=\beta (z,z)\mp i\beta
(z,Jz).
\end{equation*}
Now introduce complexification $z\leftrightarrow \mathbf{z}$ by letting $
Jz\leftrightarrow i\mathbf{z}$ and define two Hermitian forms on the
complexification $\mathbf{Z}$ of $Z:$
\begin{equation}
\boldsymbol{\beta }^{\mp }(\mathbf{z,z}^{\prime })=\beta (z,z^{\prime })\mp
i\beta (z,Jz^{\prime }).  \label{U3}
\end{equation}

Then $\boldsymbol{\beta }^{-}$ is sesquilinear i.e. complex linear with
respect to $\mathbf{z}^{\prime }$ and anti-linear with respect to $\mathbf{z,
}$ while $\boldsymbol{\beta }^{+}$ is anti-sesquilinear. From (\ref{U3}), (
\ref{U2}),
\begin{equation}
\boldsymbol{\beta }^{\mp }\mathbf{(z,z)}=\mathrm{Re}\beta (z,z)\pm \mathrm{Im
}\beta (z,Jz)\geq 0\quad \text{for all }\mathbf{z\in Z,}  \label{U4}
\end{equation}
hence by (anti-)sesquilinearity
\begin{equation*}
\sum_{jk}\bar{c}_{j}c_{k}\boldsymbol{\beta }^{\mp }\left( \mathbf{z}_{j},
\mathbf{z}_{k}\right) \geq 0.
\end{equation*}
By adding the two inequalities corresponding to plus and minus, we get (\ref
{U1}).

Conversely, (\ref{U1})$\Longrightarrow $(\ref{U2}): Applying (\ref{U1}) to
the collection $\left\{ z_{j},Jz_{j}\right\} \subset Z,\left\{ c_{j},\pm
ic_{j}\right\} \subset \mathbb{C}$ we obtain
\begin{equation*}
\sum_{jk}\bar{c}_{j}c_{k}\left[ \beta \left( z_{j},z_{k}\right) \pm i\beta
\left( z_{j},Jz_{k}\right) \right] \geq 0,
\end{equation*}
hence the forms (\ref{U3}) are nonnegative definite. By
(anti-)sesquilinearity of these forms, this is equivalent to (\ref{U4}) i.e.
(\ref{U2}).
\end{proof}

Assume that in $Z_{A},Z_{B}$ operators of complex structure $J_{A},J_{B}$
are fixed, and let $U_{\phi }^{A},U_{\phi }^{B}$ be the corresponding gauge
operators in $\mathcal{H}_{A},\mathcal{H}_{B}$ acting according (\ref
{gaugetr}). Channel $\Phi : A\rightarrow B$ is called \textit{gauge-covariant
}, if
\begin{equation}
\Phi [ U_{\phi }^{A}\rho \left( U_{\phi }^{A}\right) ^{\ast }]=U_{\phi
}^{B}\Phi [ \rho ]\left( U_{\phi }^{B}\right) ^{\ast }  \label{calinvch}
\end{equation}
for all input states $\rho $ and all $\phi \in [ 0,2\pi ].$ For the Gaussian
channel (\ref{bosgaus}) with parameters $(K,l,\mu )$ this reduces to
\begin{equation*}
l(z)\equiv 0,\quad KJ_{B}-J_{A}K=0,\quad \mu (z,z^{\prime })=\mu
(J_{B}z,J_{B}z^{\prime }).
\end{equation*}
The relation (\ref{bosgaus}) for gauge-covariant Gaussian channel takes the
form
\begin{equation}
\Phi ^{\ast }[D_{B}(\mathbf{z})]=D_{A}(\mathbf{Kz})\exp \left[ -
\boldsymbol{\mu }(\mathbf{z},\mathbf{z})\right] ,\quad \mathbf{z}\in \mathbf{
Z}_{B},  \label{g_co}
\end{equation}
where
\begin{equation}
\boldsymbol{\mu }\geq \pm \frac{1}{2}\left[ \mathbf{j}_{B}-\mathbf{K}^{\ast }
\mathbf{j}_{A}\mathbf{K}\right]  \label{n-s_channel_gauge}
\end{equation}

The equivalence of (\ref{n-s_channel_gauge}) and (\ref{n-s channel}) is
obtained by applying the lemma \ref{L0} to the Hermitian form
\begin{equation*}
\beta (z,z^{\prime })=\mu (z,z^{\prime })-\frac{i}{2}\left[ \Delta
_{B}(z,z^{\prime })-\Delta _{A}(Kz,Kz^{\prime })\right] .
\end{equation*}

Channel $\Phi : A\rightarrow B$ is called \textit{gauge-contravariant}, if
\begin{equation}
\Phi [ U_{\phi }^{A}\rho \left( U_{\phi }^{A}\right) ^{\ast }]=\left(
U_{\phi }^{B}\right) ^{\ast }\Phi [ \rho ]U_{\phi }^{B}
\end{equation}
for all input states $\rho $ and all $\phi \in [ 0,2\pi ].$ For the Gaussian
channel (\ref{bosgaus}) with parameters $(K,l,\mu )$ this reduces to
\begin{equation*}
l(z)\equiv 0,\quad KJ_{B}+J_{A}K=0,\quad \mu (z,z^{\prime })=\mu
(J_{B}z,J_{B}z^{\prime }).
\end{equation*}
The relation (\ref{bosgaus}) for gauge-contravariant Gaussian channel takes
the form
\begin{equation}
\Phi ^{\ast }[D_{B}(\mathbf{z})]=D_{A}(-\Lambda \mathbf{Kz})\exp \left[ -\boldsymbol{\mu }(\mathbf{z},\mathbf{z})\right] ,\quad \mathbf{z}
\in \mathbf{Z}_{B},  \label{g_contra}
\end{equation}
where $\Lambda $ is antilinear operator of complex conjugation, $\Lambda
^{2}=I,\,\Lambda ^{s}=-\Lambda $ in $\mathbf{Z}_{A}$ such that $\Lambda
J_{A}+J_{A}\Lambda =0,$ and $\mathbf{K=-}\Lambda K$ is complex linear
operator from $\mathbf{Z}_{B}$ to $\mathbf{Z}_{A}.$ Here
\begin{equation}
\boldsymbol{\mu }\geq \pm \frac{1}{2}\left[ \mathbf{j}_{B}+\mathbf{K}^{\ast }
\mathbf{j}_{A}\mathbf{K}\right] .  \label{n-s_channel_contra}
\end{equation}
The last condition is obtained by applying Lemma \ref{L0} to the Hermitian
form
\begin{eqnarray*}
\beta (z,z^{\prime }) &=&\mu (z,z^{\prime })-\frac{i}{2}\left[ \Delta
_{B}(z,z^{\prime })-\Delta _{A}(Kz,Kz^{\prime })\right] \\
&=&\frac{i}{2}\left[ \Delta _{B}(z,z^{\prime })+\Delta _{A}(\mathbf{K}z,
\mathbf{K}z^{\prime })\right] .
\end{eqnarray*}

\subsection{Attenuators and amplifiers}

\label{mlt}

In what follows we restrict to channels that are gauge-covariant or
contravariant with respect to fixed complex structures. Therefore, to be
specific, we consider vectors in $\mathbf{Z}$ as $s-$dimensional complex
column vectors, where the operator $J$ acts as multiplication by $i$, the
corresponding Hermitian inner product is $\mathbf{j}(\mathbf{z},\mathbf{z}
^{\prime })=\mathbf{z}^{\ast }\mathbf{z}^{\prime }$ and the symplectic form
is $\Delta (z,z^{\prime })=2\mathrm{Im}\mathbf{z}^{\ast }\mathbf{z}^{\prime },
$ where $^{\ast }$ denotes Hermitian conjugation. The linear operators in $
\mathbf{Z}$ commuting with $J$ are represented by complex $s\times s-$
matrices. The gauge group acts in $\mathbf{Z}$ as multiplication by $
e^{i\phi }$. Gaussian gauge-invariant states are described by the modified
characteristic function
\begin{equation}
\phi (\mathbf{z})=\mathrm{Tr}\rho D(\mathbf{z})=\exp \left( -\mathbf{z}
^{\ast }\boldsymbol{\alpha z}\right) ,  \label{gausstate}
\end{equation}
where $\boldsymbol{\alpha }$ is a Hermitian correlation matrix satisfying $
\boldsymbol{\alpha }\geq \mathbf{I}/2$ as follows from (\ref{n-s_comp}). For
the given complex structure, the unique minimal solution of the last
inequality is $\frac{1}{2}\mathbf{I,}$ to which correspond the vacuum state $
\rho _{0}$ and the family of coherent states $\left\{ \rho _{\mathbf{z}};\,
\mathbf{z\in Z}\right\} ,$ such that $\rho _{\mathbf{z}}=D(\mathbf{z})\rho
_{0}D(\mathbf{z})^{\ast }.$ One has
\begin{equation*}
\mathrm{Tr}\rho _{\mathbf{w}}D(\mathbf{z})=\exp \left( 2i\,\mathrm{Im}
\mathbf{w}^{\ast }\mathbf{z}-\frac{1}{2}|\mathbf{z}|^{2}\right) ,
\end{equation*}
where $|\mathbf{z}|^{2}=\mathbf{z}^{\ast }\mathbf{z}.$

Let $\mathbf{Z}_{A},\mathbf{Z}_{B}$ be the input and output spaces of
dimensionalities $s_{A},s_{B}.$ We denote by $s_{A}=\dim \mathbf{Z}_{A}$, $
s_{B}=\dim \mathbf{Z}_{B}$ the numbers of modes of the input and output of
the channel. The action of a Gaussian gauge-covariant channel (\ref{g_co})
can be described as
\begin{equation}
\Phi ^{\ast }[D_{B}(\mathbf{z})]=D_{A}(\mathbf{Kz})\exp \left( -\mathbf{z}
^{\ast }\boldsymbol{\mu z}\right) ,\quad \mathbf{z}\in\mathbf{Z}_{B},
\label{defprima}
\end{equation}
where $\mathbf{K}$ is complex $s_{B}\times s_{A}-$matrix, $\boldsymbol{\mu }$
is Hermitian $s_{B}\times s_{B}-$matrix satisfying the condition (see \cite
{htw})
\begin{equation}
\boldsymbol{\mu} \geq \pm \frac{1}{2}\left( \mathbf{I}_{B}-\mathbf{K^{\ast }K
}\right) ,  \label{n-s_channel_cov}
\end{equation}
where $\mathbf{I}_{B}$ is the unit $s_{B}\times s_{B}-$matrix. This follows
from (\ref{n-s_channel_gauge}) by taking into account that the matrix of the
form $\mathbf{j}(\mathbf{z},\mathbf{z}^{\prime })$ in an orthonormal basis
is just the unit matrix $\mathbf{I}$ of the corresponding size. Later we
will need the following

\begin{lemma}
\label{L_invert} The map (\ref{defprima}) is injective if and only if
\footnote{
For Hermitian matrices $M,N,$ the strict inequality $M>N$ means that $M-N$
is positive definite.} $\mathbf{KK^{\ast }>0}$ (in which case necessarily $
s_{B}\geq s_{A}$).
\end{lemma}

\begin{proof}
Injectivity means that $\Phi [ \rho _{1}]=\Phi [ \rho _{2}]$ implies $\rho
_{1}=\rho _{2}.$ But $\Phi [ \rho _{1}]=\Phi [ \rho _{2}]$ is equivalent to $
\mathrm{Tr}\rho _{1}\Phi ^{\ast }[D_{B}(\mathbf{z})]=\mathrm{Tr}\rho
_{2}\Phi ^{\ast }[D_{B}(\mathbf{z})],$ i.e. $\mathrm{Tr}\rho _{1}D_{A}(
\mathbf{Kz})=\mathrm{Tr}\rho _{2}D_{A}(\mathbf{Kz})$ for all $\mathbf{z\in Z}
_{B}.$ By irreducibility of the Weyl system, this property is equivalent to
\textrm{Ran}$\mathbf{K}=\mathbf{Z}_{A},$ i.e. \textrm{Ker}$
\mathbf{K^{\ast }=\{0\}}$ or $\mathbf{KK}^{\ast }>0.$
\end{proof}

The channel (\ref{defprima}) is extreme if $\mu $ is a minimal solution of
the inequality (\ref{n-s_channel_cov}). Special cases of the maps~(\ref
{defprima}) are provided by the \textit{attenuator} and \textit{amplifier}
channels, characterized by matrix $\mathbf{K}$ fulfilling the inequalities, $
\mathbf{K^{\ast }K\leq I}$ and $\mathbf{K^{\ast }K\geq I}$ respectively. We
are particularly interested in \textit{extreme attenuator} which corresponds
to
\begin{equation}
\mathbf{K^{\ast }K}\leq \mathbf{I}_{B},\qquad \qquad \boldsymbol{\mu} =\frac{
1}{2}\left( \mathbf{I}_{B}-\mathbf{K} ^{\ast }\mathbf{K}\right) ,
\label{mindef1}
\end{equation}
and \textit{extreme amplifier}
\begin{equation}
\mathbf{K^{\ast }K}\geq \mathbf{I}_{B},\qquad \qquad \boldsymbol{\mu }=\frac{
1}{2}\left( \mathbf{\ K^{\ast }K}-\mathbf{I}_{B}\right) .  \label{mindef2}
\end{equation}

Denoting by $\mathbf{\bar{z}}$ the column vector obtained by taking the
complex conjugate of the elements of $\mathbf{z,}$ the action of the
Gaussian gauge-contravariant channel (\ref{g_contra}) is described as
\begin{equation}
\Phi ^{\ast }[D_{B}(\mathbf{z})]=D_{A}(-\overline{\mathbf{Kz}})\exp \left( -
\mathbf{z}^{\ast }\boldsymbol{\mu z}\right) ,  \label{CONTRAV}
\end{equation}
where $\boldsymbol{\mu }$ is Hermitian matrix satisfying the inequality
\begin{equation}
\boldsymbol{\mu} \geq \frac{1}{2}\left( \mathbf{I}_{B}+\mathbf{K}^{\ast }
\mathbf{K}\right) ,  \label{ineq2}
\end{equation}
which follows from (\ref{n-s_channel_contra}). Here $\bar{\mathbf{z}}$ is
the column vector consisting of complex conjugates of the components of $
\mathbf{z}$. These maps are extreme if
\begin{equation}
\boldsymbol{\mu }=\frac{1}{2}\left( \mathbf{I}_{B}+\mathbf{K^{\ast }K}
\right) .  \label{mindef5}
\end{equation}

The following proposition generalizes to many modes the decomposition of
one-mode channels the usefulness of which was emphasized and exploited in
the paper \cite{gp} (see also \cite{cgh} on concatenations of one-mode
channels):

\begin{proposition}
\label{prop1} Any Gaussian gauge-covariant channel $\Phi :A\rightarrow B$ is
a concatenation $\Phi =\Phi _{2}\circ $ $\Phi _{1}$ \ of extreme attenuator $
\Phi _{1}:A\rightarrow B$ and extreme amplifier $\Phi _{2}:B\rightarrow B$.

Any Gaussian gauge-contravariant channel $\Phi :A\rightarrow B$ is a
concatenation of extreme attenuator $\Phi _{1}:A\rightarrow B$ and extreme
gauge-contravariant channel $\Phi _{2}:B\rightarrow B$.
\end{proposition}

\begin{proof}
The concatenation $\Phi =\Phi _{2}\circ \Phi _{1}$ of Gaussian
gauge-covariant channels $\Phi _{1}$ and $\Phi _{2}$ obeys the rule:
\begin{eqnarray}
\mathbf{K} &=&\mathbf{K}_{1}\mathbf{K}_{2},\quad  \label{cc1} \\
\boldsymbol{\mu } &=&\mathbf{K}_{2}^{\ast }\boldsymbol{\mu }_{1}\mathbf{K}
_{2}+\boldsymbol{\mu }_{2}.  \label{cc2}
\end{eqnarray}
By inserting relations
\begin{equation*}
\boldsymbol{\mu }_{1}=\frac{1}{2}\left( \mathbf{I}_{B}-\mathbf{K}_{1}^{\ast }
\mathbf{K}_{1}\right) =\frac{1}{2}\left( \mathbf{I}_{B}-|\mathbf{K}|_{1}
^{2}\right),\quad \boldsymbol{\mu }_{2}=\frac{1}{2}\left( \mathbf{K}
_{2}^{\ast }\mathbf{K}_{2}-\mathbf{I}_{B}\right) =\frac{1}{2}\left(
\left\vert \mathbf{K}_{2}\right\vert ^{2}-\mathbf{I}_{B}\right)
\end{equation*}
into (\ref{cc2}) and using (\ref{cc1}) we obtain
\begin{equation}
\left\vert \mathbf{K}_{2}\right\vert ^{2}=\mathbf{K}_{2}^{\ast }\mathbf{K}
_{2}=\boldsymbol{\mu} +\frac{1}{2}(\mathbf{K}^{\ast }\mathbf{K}+\mathbf{I}
_{B})\geq \left\{
\begin{array}{c}
\mathbf{I}_{B} \\
\mathbf{K}^{\ast }\mathbf{K}
\end{array}
\right.  \label{ineq1}
\end{equation}
from the inequality (\ref{n-s_channel_cov}). By using operator monotonicity
of the square root, we have
\begin{equation*}
\left\vert \mathbf{K}_{2}\right\vert \geq\mathbf{\ I}_{B},\quad\left\vert
\mathbf{K}_{2}\right\vert \geq \left\vert \mathbf{K}\right\vert .
\end{equation*}
The first inequality (\ref{ineq1}) implies that choosing
\begin{equation}
\mathbf{K}_{2}=\left\vert \mathbf{K}_{2}\right\vert =\sqrt{ \boldsymbol{\mu}
+\frac{1}{2}(\mathbf{K}^{\ast }\mathbf{K}+\mathbf{I}_{B})}  \label{K_2}
\end{equation}
and the corresponding $\boldsymbol{\mu }_{2}=\frac{1}{2}\left( \left\vert
\mathbf{K}_{2}\right\vert ^{2}-\mathbf{I}_{B}\right) ,$ we obtain extreme
amplifier $\Phi _{2}:B\rightarrow B$.

Then with
\begin{equation}
\mathbf{K}_{1}\mathbf{=K}\left\vert \mathbf{K}_{2}\right\vert ^{-1}
\label{K_1}
\end{equation}
we obtain, taking into account the second inequality in (\ref{ineq1}) and
also Lemma \ref{L7} below,
\begin{equation}
\mathbf{K}_{1}\mathbf{K}_{1}^{\ast }\mathbf{=K}\left\vert \mathbf{K}
_{2}\right\vert ^{-2}\mathbf{K}^{\ast }\mathbf{=K}\left[ \boldsymbol{\mu +}
\frac{ 1}{2}\mathbf{(\mathbf{K}^{\ast }K+I)}\right] ^{-1}\mathbf{K}^{\ast }
\leq \mathbf{I}_{A},  \label{K1}
\end{equation}
which implies $\mathbf{K}_{1}^{\ast }\mathbf{K}_{1}\leq\mathbf{\ I}_{A},$
hence $\mathbf{K}_{1}$ with the corresponding $\boldsymbol{\mu }_{1}=\frac{1
}{2}\left( \mathbf{I}_{B}- \mathbf{K}_{1}^{\ast }\mathbf{K}_{1}\right) $
give the quantum-limited attenuator.

\begin{lemma}
\label{L7} Let $\mathbf{M}\geq\mathbf{K}^{\ast }\mathbf{K}$, then $\mathbf{K}
\mathbf{M}^{-}\mathbf{K}^{\ast }\leq\mathbf{I}_{A}$, where $^{-}$ means
(generalized) inverse.
\end{lemma}

\begin{proof}
By the definition of the generalized inverse,
\begin{equation*}
u^{\ast }\mathbf{M}^{-}u=\sup_{v: v\in\mathrm{Ran}\mathbf{M}, v\neq 0}
\frac{|u^{\ast }v|^2}{v^{\ast }\mathbf{M}v}.
\end{equation*}
By inserting $\mathbf{K}^{\ast }u$ in place of $u$ and using Cauchy-Schwarz
inequality in the nominator of the fraction, we obtain
\begin{equation*}
u^{\ast }\mathbf{K}\mathbf{M}^{-}\mathbf{K}^{\ast }u\leq\sup_{v: v\in\mathrm{
Ran}\mathbf{M}, v\neq 0}\frac{ u^{\ast }u\,v^{\ast }\mathbf{K}^{\ast }
\mathbf{K}v}{v^{\ast }\mathbf{M}v}\leq u^{\ast }u.
\end{equation*}
\end{proof}

In the case of contravariant channel the relations (\ref{cc1}), (\ref{cc2})
are replaced with
\begin{eqnarray}
\mathbf{\bar{K}} &=&\mathbf{K}_{1}\mathbf{\bar{K}}_{2},  \label{bc1} \\
\boldsymbol{\mu } &=&\mathbf{K}_{2}^{\ast }\bar{\boldsymbol{\mu}}_{1}\mathbf{
K}_{2}+\boldsymbol{\mu }_{2}.  \label{bc2}
\end{eqnarray}
By substituting
\begin{equation*}
\boldsymbol{\mu }_{1}\mathbf{=}\frac{1}{2}\left( \mathbf{I}-\mathbf{K}
_{1}^{\ast }\mathbf{K}_{1}\right) \mathbf{,\quad \mu }_{2}=\frac{1}{2}\left(
\mathbf{K}_{2}^{\ast }\mathbf{K}_{2}+\mathbf{I}_{B}\right)
\end{equation*}
into (\ref{bc2}) and using (\ref{ineq2}) we obtain
\begin{equation}
\left\vert \mathbf{K}_{2}\right\vert ^{2}=\mathbf{K}_{2}^{\ast }\mathbf{K}
_{2}=\boldsymbol{\mu} +\frac{1}{2}(\mathbf{K}^{\ast }\mathbf{K}-\mathbf{I}
_{B})\geq \mathbf{K} ^{\ast }\mathbf{K}.  \label{K2}
\end{equation}
Taking $\mathbf{K}_{2}\mathbf{=}\left\vert \mathbf{K}_{2}\right\vert ,$ $
\boldsymbol{\mu }_{2}=\frac{1}{2}\left( |\mathbf{K}_{2}|^{2}+\mathbf{I}_{B}
\right) $ gives extreme gauge-contravariant channel $\Phi _{2}:B\rightarrow
B $ . With
\begin{equation}
\mathbf{\bar{K}}_{1}=\mathbf{K}\left\vert \mathbf{K}_{2}\right\vert ^{-}
\label{Kbar_1}
\end{equation}
we obtain, by using Lemma \ref{L7},
\begin{eqnarray}
\mathbf{\bar{K}}_{1}\mathbf{\bar{K}}_{1}^{\ast } &=&\mathbf{K}\left(
\left\vert \mathbf{K}_{2}\right\vert ^{-}\right) ^{2}\mathbf{K}^{\ast }
\notag \\
&=&\mathbf{K}\left[ \boldsymbol{\mu} +\frac{1}{2}(\mathbf{K}^{\ast }\mathbf{K
}-\mathbf{I}_{B}) \right] ^{-}\mathbf{K}^{\ast }\leq \mathbf{I}_{A},
\label{K2bar}
\end{eqnarray}
which implies $\mathbf{K}_{1}\mathbf{K}_{1}^{\ast }\leq \mathbf{I}_{A},$
with the corresponding $\boldsymbol{\mu }_{1}$ give the extreme attenuator $
\Phi _{1}:A\rightarrow B$.
\end{proof}

\begin{remark}
\label{rem} In the case of gauge-covariant channel, the equality in (\ref{K1}
) shows that $ \mathbf{K}\mathbf{K}^{\ast }>0$ implies $\mathbf{K}_{1}
\mathbf{K}_{1}^{\ast }>0$, while the inequality $\boldsymbol{\mu} >\frac{1}{2
}\left( \mathbf{K}^{\ast }\mathbf{K}-\mathbf{I}_{B}\right) $ implies $
\mathbf{K} _{1}\mathbf{K}_{1}^{\ast }<\mathbf{I}_{A}$. In the case of
gauge-contravariant channel, the inequality $\boldsymbol{\mu} >\frac{1}{2}
\left( \mathbf{I}_{B}+\mathbf{K}^{\ast }\mathbf{K}\right) $ implies $0<
\mathbf{K}_{1}\mathbf{K}_{1}^{\ast }<\mathbf{I}_{A}$ via (\ref{K2bar}).
\end{remark}

\begin{proposition}
\label{prop9} The extreme attenuator with matrix $\mathbf{K}$ and extreme
attenuator with matrix $\tilde{\mathbf{K}}=\sqrt{\mathbf{I}_{A}-\mathbf{KK}
^{\ast }}$ are mutually complementary.

The extreme amplifier with matrix \ $\mathbf{K}$ and gauge-contravariant
channel with matrix $\tilde{\mathbf{K}}=\sqrt{\mathbf{\bar{K}\bar{K}}^{\ast
}-\mathbf{I}_{A}}$ are mutually complementary.
\end{proposition}

\begin{proof}
For the case of one mode see \cite{cgh} or \cite{h}, Sec. 12.6.1. We sketch
the proof for several modes below. Define $\mathbf{Z}_{E}\simeq \mathbf{Z}
_{A},\,\mathbf{Z}_{D}\simeq \mathbf{Z}_{B},$ so that $\mathbf{Z=Z}_{A}\oplus
\mathbf{Z}_{D}\simeq \mathbf{Z}_{B}\oplus \mathbf{Z}_{E}$

In the case of attenuator consider the block unitary matrix in $\mathbf{Z}:$
\begin{equation}
\mathbf{V}=\left[
\begin{array}{cc}
\mathbf{K} & \sqrt{\mathbf{I}_{A}-\mathbf{KK}^{\ast }} \\
\sqrt{\mathbf{I}_{B}-\mathbf{K}^{\ast }\mathbf{K}} & -\mathbf{K}^{\ast }
\end{array}
\right]  \label{V}
\end{equation}
which defines unitary dynamics $U$ in $\mathcal{H}=\mathcal{H}_{A}\otimes
\mathcal{H}_{D}\simeq \mathcal{H}_{B}\otimes \mathcal{H}_{E}$ by the
relation $U^{\ast }D_{BE}(\mathbf{z}_{BE})U=D_{AD}(\mathbf{Vz}_{BE})$. Here $
\mathbf{z}_{BE}=[\mathbf{z}_{B}\,\,\mathbf{z}_{E}]^{t},$ $D_{BE}(\mathbf{z}
_{BE})=D_{B}(\mathbf{\ z}_{B})\otimes D_{E}(\mathbf{z}_{E})$, and the
unitarity follows from the relation
\begin{equation}
\mathbf{K}\sqrt{\mathbf{I}_{B}-\mathbf{K}^{\ast }\mathbf{K}}=\sqrt{\mathbf{I}
_{A}- \mathbf{KK}^{\ast }}\mathbf{K}.  \label{kkk}
\end{equation}
Let $\rho _{D}=\rho _{0}$ be the vacuum state, $\rho _{A}=\rho$ an arbitrary
state. Then the formulas (\ref{fia}), (\ref{fie}) define the mutually
complementary extreme attenuators as described in the first statement. The
proof is obtained by computing the characteristic function of the output
states for the channels. For the state of the composite system $\rho
_{BE}=U(\rho \otimes \rho _{D})U^{\ast }$ we have
\begin{eqnarray}
\phi _{BE}(\mathbf{z}_{BE}) &=&\mathrm{Tr}U(\rho \otimes \rho _{D})U^{\ast }
\left[ D_{B}(\mathbf{z}_{B})\otimes D_{E}(\mathbf{z}_{E})\right]  \notag \\
&=&\mathrm{Tr}(\rho \otimes \rho _{D})U^{\ast }\left[ D_{B}(\mathbf{z}
_{B})\otimes D_{E}(\mathbf{z}_{E})\right] U  \notag \\
&=&\mathrm{Tr}(\rho \otimes \rho _{D})\left[ D_{A}(\mathbf{Kz_{B}+\tilde{K}z}
_{E})\otimes D_{D}(\mathbf{\sqrt{\mathbf{I}_{B}-\mathbf{K}^{\ast }\mathbf{K}}
z_{B}-K^{\ast }z}_{E})\right]  \notag \\
&=&\phi _{A}(\mathbf{Kz_{B}+\tilde{K}z}_{E})\exp \left[ -\frac{1}{2}|\mathbf{
\sqrt{\mathbf{I}_{B}-\mathbf{K}^{\ast }\mathbf{K}}z_{B}}-\mathbf{K^{\ast }z}
_{E}|^{2}\right] .  \label{hren}
\end{eqnarray}
By setting $\mathbf{z}_{E}=0$ or $\mathbf{z}_{B}=0$ we obtain
\begin{eqnarray*}
\phi _{B}(\mathbf{z}_{B}) &=&\phi _{A}(\mathbf{Kz_{B}})\exp \left[ -\frac{1}{
2}\mathbf{z}_{B}^{\ast }\left( \mathbf{I}_{B}-\mathbf{K}^{\ast }\mathbf{K}
\right) \mathbf{z}_{B}\right] , \\
\phi _{E}(\mathbf{z}_{E}) &=&\phi _{A}(\mathbf{\tilde{K}z}_{E})\exp \left[ -
\frac{1}{2}\mathbf{\mathbf{z}_{E}^{\ast }KK^{\ast }z}_{E}\right]
\end{eqnarray*}
as required.

In the case of amplifier, set
\begin{equation*}
\mathbf{V}=\left[
\begin{array}{cc}
\mathbf{K} & -\sqrt{\mathbf{KK}^{\ast }-\mathbf{I}_{A}}\Lambda  \\
-\Lambda \sqrt{\mathbf{K}^{\ast }\mathbf{K}-\mathbf{I}_{B}} & \Lambda
\mathbf{K}^{\ast }\Lambda
\end{array}
\right] ,
\end{equation*}
where $\Lambda $ is the operator of complex conjugation,
anticommuting with multiplication by $i.$ By using the property $\Delta
(\Lambda z,\Lambda z^{\prime })=-\Delta (z,z^{\prime })$, we obtain that $\mathbf{V}$
corresponds to a symplectic transformation in $\mathbf{Z}$ generating
unitary dynamics $U$ in $\mathcal{H}$. Let again $\rho _{0}$ be the vacuum
state of the environment. Then the formulas (\ref{fia}), (\ref{fie}) define
the mutually complementary channels as described in the second statement of
the Proposition, and the proof is similar.

To show that $\mathbf{V}$ is a symplectic transformation, introduce the
matrices
\begin{equation*}
\mathbf{\Theta }=\left[
\begin{array}{cc}
\mathbf{I}_{B} & 0 \\
0 & -\Lambda
\end{array}
\right] ,\quad \mathbf{V}_{1}=\left[
\begin{array}{cc}
\mathbf{K} & \sqrt{\mathbf{KK}^{\ast }-\mathbf{I}_{A}} \\
\sqrt{\mathbf{K}^{\ast }\mathbf{K}-\mathbf{I}_{B}} & \mathbf{K}^{\ast }
\end{array}
\right] ,\quad \mathbf{\Sigma}=\left[
\begin{array}{cc}
\mathbf{I}_{B} & 0 \\
0 & -\mathbf{I}_{A}
\end{array}
\right].
\end{equation*}
Notice that $\mathbf{V}=\mathbf{\Theta} \mathbf{V}_{1}\mathbf{\Theta },$  and $
\mathbf{V}_{1}^{\ast }\mathbf{\Sigma}\mathbf{V}_{1}=\mathbf{\Sigma}$, which means that $\mathbf{V}_{1}$
preserves the indefinite Hermitian form $\sigma (z_{BE}, z^{\prime}_{BE})=z^{\ast}_{B}z^{\prime}_{B}-z^{\ast}_{E}z^{\prime}_{E}$.
By taking into account
\begin{equation*}
  \Delta_{BE}(\mathbf{\Theta}z_{BE}, \mathbf{\Theta}z^{\prime}_{BE})=\Delta_{B}(z_{B}, z^{\prime}_{B})-\Delta_{E}(z_{E}, z^{\prime}_{E})=\mathrm{Im}\sigma (z_{BE}, z^{\prime}_{BE}),
\end{equation*}
we obtain
\begin{equation*}
\Delta_{BE}(\mathbf{V}z_{BE}, \mathbf{V}z^{\prime}_{BE})=\mathrm{Im}\sigma (\mathbf{V}_{1}\mathbf{\Theta }z_{BE}, \mathbf{V}_{1}\mathbf{\Theta }z^{\prime}_{BE})
=\mathrm{Im}\sigma (\mathbf{\Theta }z_{BE}, \mathbf{\Theta }z^{\prime}_{BE})=\Delta_{BE}(z_{BE}, z^{\prime}_{BE}),
\end{equation*}
as required.
\end{proof}

Again, later we will need the following

\begin{lemma}
\label{L5} Let $\Phi _{1}:A\rightarrow B$ be an extreme attenuator with $
\mathbf{0}<\mathbf{K}_{1}\mathbf{K}_{1}^{\ast }<\mathbf{I}_{A},$ then $\Phi
_{1}[P_{\psi }]=$ $P_{\psi ^{\prime }}$ (a pure state) if and only if $
P_{\psi }$ is a coherent state.
\end{lemma}

\begin{proof}
According to Proposition \ref{prop9}, the complementary channel $\tilde{\Phi}
_{1}$ is an extreme attenuator with the matrix $\tilde{\mathbf{K}}=\sqrt{
\mathbf{I}_{A}-\mathbf{K}_{1}\mathbf{K}_{1}^{\ast }},$ such that $\mathbf{0}<
\mathbf{\tilde{K}}<\mathbf{I}_{A}.$ Its output is also pure, $\Phi
_{1}[P_{\psi }]=P_{\psi _{E}^{\prime }},$ as the outputs of complementary
channels have identical nonzero spectra by Lemma \ref{L1}. Thus
\begin{equation*}
U(\psi \otimes \psi _{0})=\psi ^{\prime }\otimes \psi _{E}^{\prime },
\end{equation*}
where $\psi _{0}\in \mathcal{H}_{D}$ is the vacuum vector and $U$ is the
unitary operator in $\mathcal{H}$ implementing the symplectic transformation
corresponding to the unitary (\ref{V}) in $\mathbf{Z}_{A}\mathbf{\oplus Z}
_{D}\simeq \mathbf{Z}_{B}\mathbf{\oplus Z}_{E},$ with $\mathbf{Z}_{D}\simeq
\mathbf{Z}_{B},\,\mathbf{Z}_{E}\simeq \mathbf{Z}_{A}$ Denoting by
\begin{equation*}
\phi (\mathbf{z})=\mathrm{Tr}P_{\psi }D_{A}(\mathbf{z}),\quad \,\phi
^{\prime }(\mathbf{z}_{B})=\mathrm{Tr}P_{\psi ^{\prime }}D_{B}(\mathbf{z}
_{B}),\quad \phi _{E}(\mathbf{z}_{E})=\mathrm{Tr}P_{\psi _{E}^{\prime
}}\,D_{E}(\mathbf{z}_{E})
\end{equation*}
the quantum characteristic functions and using the relation (\ref{hren}), we
have the functional equation
\begin{equation}
\phi ^{\prime }(\mathbf{z}_{B})\phi _{E}(\mathbf{z}_{E})=\phi (\mathbf{K}_{1}
\mathbf{z} _{B}+\tilde{\mathbf{K}}\mathbf{z}_{E})\exp \left[ -\frac{1}{2}|
\sqrt{\mathbf{I}_{B} -\mathbf{K}_{1}^{\ast }\mathbf{K}_{1}}\mathbf{z}_{B}-
\mathbf{K}_{1}^{\ast }\mathbf{z} _{E}|^{2}\right] .  \label{FE}
\end{equation}
By letting $\mathbf{z}_{E}=0,$ respectively $\mathbf{z}=0,$ we obtain
\begin{eqnarray*}
\phi ^{\prime }(\mathbf{z}_{B}) &=&\phi (\mathbf{K}_{1}\mathbf{z}_{B})\exp
\left[ -\frac{1}{2}|\sqrt{\mathbf{I}_{B}-\mathbf{K}_{1}^{\ast }\mathbf{K}_{1}
}\mathbf{z}_{B}|^{2}\right] , \\
\quad \phi _{E}(\mathbf{z}_{E}) &=&\phi (\tilde{\mathbf{K}}\mathbf{z}
_{E})\exp \left[ -\frac{1}{2}|\mathbf{K}_{1}^{\ast }\mathbf{z}_{E}|^{2}
\right] ,
\end{eqnarray*}
thus, after the change of variables $\mathbf{z}=\mathbf{K}_{1}\mathbf{z}
_{B},\,\mathbf{z}^{\prime }=\tilde{\mathbf{K}}\mathbf{z}_{E}$, and using (
\ref{kkk}), the equation (\ref{FE}) reduces to
\begin{equation*}
\phi (\mathbf{z})\phi (\mathbf{z}^{\prime })=\phi (\mathbf{z}+\mathbf{z}
^{\prime })\exp \left[ \mathrm{Re}\,\mathbf{z}^{\ast }\mathbf{z}^{\prime }
\right] .
\end{equation*}
The condition of the Lemma ensures that \textrm{Ran\thinspace }$\mathbf{
K_{1}=\,}$\textrm{Ran\thinspace }$\tilde{\mathbf{K}}=\mathbf{Z}_{A}$.
Substituting $\omega (\mathbf{z})=\phi (\mathbf{z})\exp \left[ \frac{1}{2}|
\mathbf{z}|^{2}\right] ,$ this becomes
\begin{equation}
\omega (\mathbf{z})\omega (\mathbf{z}^{\prime })=\omega (\mathbf{z}+\mathbf{z
}^{\prime })  \label{FE1}
\end{equation}
for all $\mathbf{z}, \mathbf{z}^{\prime }\in \mathbf{Z}_{A}$. The function $
\omega (\mathbf{z}),$ as well as the characteristic function $\phi (\mathbf{z
}),$ is continuous and satisfies $\omega (-\mathbf{z})=\overline{\omega (
\mathbf{z})}.$ The only solution of (\ref{FE1}) satisfying these conditions
is the exponent $\omega (\mathbf{z})=\exp \left[ i\mathrm{Im}\mathbf{w}
^{\ast }\mathbf{z}\right] $ for some complex $\mathbf{w}.$ Thus
\begin{equation*}
\phi (\mathbf{z})=\exp \left[ i\mathrm{Im\,}\mathbf{w}^{\ast }\mathbf{z}-
\frac{1}{2}|\mathbf{z}|^{2}\right]
\end{equation*}
is the characteristic function of the coherent state $\rho _{\mathbf{w}/2}.$
\end{proof}

\subsection{Gaussian optimizers}

The following basic result for one mode was obtained in \cite{mgh}. Here we
present a complete proof in the multimode case, a sketch of which was given
in \cite{mgh1}.

\begin{theorem}
\label{T1} (i) Let $\Phi $ be a gauge covariant or contravariant channel and
let $f$ be a real concave function on $[0,1],$ such that $f(0)=0,$ then
\begin{equation}
\mathrm{Tr}f(\Phi [ \rho ])\geq \mathrm{Tr}f(\Phi [ \rho _{\mathbf{w}}])=
\mathrm{Tr}f(\Phi [ \rho _{0}])  \label{main1}
\end{equation}
for all states $\rho $ and any coherent state $\rho _{\mathbf{w}}$ (the
value on the right is the same for all coherent states by the unitary
covariance property of a Gaussian channel (\ref{gauscov})).

(ii) Let $f$ be strictly concave, then equality in (\ref{main1}) is
attained only if $\rho $ is a coherent state in the following cases:

a) $s_{B}=s_{A}$ and $\Phi $ is an extreme amplifier with $\boldsymbol{\mu }=
\frac{1}{2}\left( \mathbf{K}^{\ast }\mathbf{K}-\mathbf{I}_{B}\right)>0 $;

b) $s_{B}\geq s_{A},\,$ the channel $\Phi $ is gauge-covariant with $\mathbf{
KK}^{\ast }>0$ and
\begin{equation}
\boldsymbol{\mu }>\frac{1}{2}\left( \mathbf{K}^{\ast }\mathbf{K}-\mathbf{I}
_{B} \right) ;  \label{mugr}
\end{equation}

c) $s_{B}\geq s_{A},\,$ the channel $\Phi $ is gauge-contravariant with $
\mathbf{KK}^{\ast }>0$ and $\boldsymbol{\mu >}\frac{1}{2}\left( \mathbf{I}
_{B}+\mathbf{K }^{\ast }\mathbf{K}\right) .$
\end{theorem}

\begin{proof}
(i) We first prove the inequality (\ref{main1}) for strictly concave $f.$
Then the inequality for arbitrary concave $f$ follows by the monotone
approximation $\ f(x)=\lim_{\varepsilon \downarrow 0}f_{\varepsilon }(x),$
since $f_{\varepsilon }(x)=f(x)-\varepsilon x^{2}$ are strictly concave.
Also, by concavity, it is sufficient to prove (\ref{main1}) only for $\rho
=P_{\psi }.$

By Proposition \ref{prop1}, $\Phi =\Phi _{2}\circ $ $\Phi _{1}$ \ where $
\Phi _{1}:A\rightarrow B$ is an extreme attenuator and $\Phi
_{2}:B\rightarrow B$ is either extreme amplifier or extreme
gauge-contravariant channel. Any extreme attenuator maps vacuum state into
vacuum. Indeed,
\begin{eqnarray*}
\mathrm{Tr}\Phi _{1}\left[ \rho _{0}\right] D_{B}(\mathbf{z}) &=&\mathrm{Tr}
\rho _{0}\Phi _{1}^{\ast }[D_{B}(\mathbf{z})] \\
&=&\mathrm{Tr}\rho _{0}D_{A}(\mathbf{Kz})\exp \left( -\frac{1}{2}\mathbf{z}
^{\ast }\left( \mathbf{I}_{B}-\mathbf{K^{\ast }K}\right) \mathbf{z}\right) \\
&=&\exp \left( -\frac{1}{2}|\mathbf{z}|^{2 }\right) =\mathrm{Tr}\rho
_{0}D_{B}(\mathbf{z}).
\end{eqnarray*}
Therefore $\mathrm{Tr}f(\Phi [ \rho _{0}])=\mathrm{Tr}f(\Phi _{2}[\rho
_{0}]).$ Then it is sufficient to prove (\ref{main1}) for all extreme
amplifiers and all extreme gauge-contravariant channels $\Phi _{2}$. Indeed,
assume that we have proved
\begin{equation}
\mathrm{Tr}f(\Phi _{2}[P_{\psi }])\geq \mathrm{Tr}f(\Phi _{2}[\rho _{0}]).
\label{amp}
\end{equation}
for any state vector $\psi .$ Consider the spectral decomposition $\Phi
_{1}[P_{\psi }]=\sum_{j}p_{j}P_{\phi _{j}},$ where $p_{j}>0,$ then
\begin{eqnarray}
\mathrm{Tr}f(\Phi [ P_{\psi }]) &=&\mathrm{Tr}f(\Phi _{2}[\Phi _{1}[P_{\psi
}]])  \label{decomp} \\
&\geq &\sum_{j}p_{j}\mathrm{Tr}f(\Phi _{2}[P_{\phi _{j}}])  \label{decompa}
\\
&\geq &\mathrm{Tr}f(\Phi _{2}[\rho _{0}])  \label{decomb} \\
&=&\mathrm{Tr}f(\Phi _{2}[\Phi _{1}[\rho _{0}]])=\mathrm{Tr}f(\Phi [ \rho
_{0}]).  \label{decomc}
\end{eqnarray}

Then, according to the second statement of Proposition \ref{prop9} and Lemma
\ref{L1}
\begin{equation*}
\mathrm{Tr}f(\Phi _{2}[P_{\psi }])=\mathrm{Tr}f(\tilde{\Phi}_{2}[P_{\psi }]),
\end{equation*}
where $\Phi _{2}$ is an extreme amplifier and $\tilde{\Phi}_{2}$ is an
extreme gauge-contravariant channel. Thus it is sufficient to prove (\ref
{amp}) only for an extreme amplifier $\Phi _{2}:B\rightarrow B,$ with
Hermitian matrix $\mathbf{K}_{2}\geq \mathbf{I}_{B}$.

The following result is based on a key observation by Giovannetti.
\begin{lemma}
\label{L4} For an extreme amplifier $\Phi _{2}:B\rightarrow B,$ with matrix $
\mathbf{K}_{2}\geq \mathbf{I}_{B},$ there is an extreme attenuator $\Phi
_{1}^{\prime }$ such that for all $\psi \in \mathcal{H}_{B}$
\begin{equation}
\Phi _{2}(P_{\psi })\thicksim \left( \Phi _{2}\circ \Phi _{1}^{\prime
}\right) (P_{\psi }).  \label{simi}
\end{equation}
\end{lemma}

\begin{proof}
By Proposition \ref{prop9} and Lemma \ref{L1} $\Phi _{2}(P_{\psi })\thicksim
\tilde{\Phi}_{2}(P_{\psi })$ for all $\psi \in \mathcal{H}_{B},$ where $
\tilde{\Phi}_{2}$ is extreme contravariant channel with the matrix $\tilde{
\mathbf{K}}=\sqrt{\bar{\mathbf{K}}_{2}^{2}-\mathbf{I}_{B}}.$

Define the transposition map $\mathcal{T}:B\rightarrow B$ by the relation $
\mathcal{T}[D(\mathbf{z})]=D(-\mathbf{\bar{z}}).$ The concatenation $\Phi =
\mathcal{T}\circ\tilde{\Phi}_{2}$ is a covariant Gaussian channel:
\begin{equation*}
\Phi ^{\ast }[D(\mathbf{z})]=\tilde{\Phi}_{2}^{\ast }\circ\mathcal{T}[D(
\mathbf{z})]=D(\sqrt{\mathbf{K}_{2}^{2}-\mathbf{I}_{B}}\mathbf{z})\exp
\left( -\frac{1}{2}\mathbf{z}^{\ast }\mathbf{K}_{2}^{2}\mathbf{z}\right) .
\end{equation*}
Applying decomposition from Proposition \ref{prop1}, namely the relation (
\ref{ineq1}), gives $\Phi =\Phi _{2}\circ \Phi _{1}^{\prime },$ where $\Phi
_{2}$ is the original amplifier, and $\Phi _{1}^{\prime }:B\rightarrow B$ is
another extreme attenuator with matrix $\mathbf{K}_{1}=\sqrt{\mathbf{I}_{B}-
\mathbf{K}_{2}^{-2}}$. This implies the relation (\ref{simi}).
\end{proof}

Lemma \ref{L4} and Lemma \ref{L1a} imply
\begin{equation}
\mathrm{Tr}f(\Phi _{2}(P_{\psi }))=\mathrm{Tr}f\left( \left( \Phi _{2}\circ
\Phi _{1}^{\prime }\right) (P_{\psi })\right) .  \label{ABBA}
\end{equation}
Again, consider the spectral decomposition of the density operator
\begin{equation*}
\Phi _{1}^{\prime }(P_{\psi })=\sum_{j}p_{j}^{\prime }P_{\psi _{j}},\,\quad
p_{j}^{\prime }>0.
\end{equation*}
By concavity,
\begin{equation}
\mathrm{Tr}f\left( \left( \Phi _{2}\circ \Phi _{1}^{\prime }\right) [P_{\psi
}]\right) \geq \sum_{j}p_{j}^{\prime }\mathrm{Tr}f\left( \Phi _{2}[P_{\psi
_{j}}]\right) .  \label{AAB}
\end{equation}

Since $f$ is assumed strictly concave, then $\rho \rightarrow \mathrm{Tr}
f(\Phi _{2}[\rho ])$ is strictly concave \cite{carlen}. Assuming that $
P_{\psi }$ is a minimizer for the functional (\ref{ABBA}), we conclude that $
\Phi _{2}[P_{\psi _{j}}]$ must all coincide, otherwise the above inequality
would be strict, contradicting the assumption. From Lemma \ref{L_invert} it
follows that $P_{\psi _{j}}=P_{\psi ^{\prime }}$ for all $j$ and for some $
\psi ^{\prime }\in \mathcal{H}_{B},$ hence, assuming that $P_{\psi }$ is a
minimizer, the output $\Phi _{1}[P_{\psi }]=$ $P_{\psi ^{\prime }}$ \ is a
pure state.

Since $\mathbf{K}_{1}=\sqrt{\mathbf{I}_{B}-\mathbf{K}_{2}^{-2}},$ the
condition of Lemma \ref{L5} is fulfilled if $\mathbf{K}_{2}>\mathbf{I}_{B}.$
In this case, if $P_{\psi }$ is a minimizer, the Lemma implies that $P_{\psi
}$ is a coherent state. Thus we obtain the inequality (\ref{amp}) for the
amplifier $\Phi _{2}$ with $\mathbf{K}_{2}>\mathbf{I}_{B}$ and strictly
concave $f$ . In this way we also obtain the case a) of the ``only if''
statement (ii).

In the case of amplifier $\Phi _{2}$ with $\mathbf{K}_{2}\geq \mathbf{I}_{B},
$ we can take any sequence $\mathbf{K}_{2}^{(n)}>\mathbf{I}_{B}$, $\mathbf{K}
_{2}^{(n)}\rightarrow\mathbf{K}_{2},$ and the corresponding amplifiers $\Phi
_{2}^{(n)}$. Then $\mathrm{Tr}f(\Phi _{2}^{(n)}[\rho ])\rightarrow \mathrm{Tr
}f(\Phi _{2}[\rho ])$ for any concave polygonal function $f$ on $[0,1],$
such that $f(0)=0,$ and any $\rho \in \mathfrak{S}(\mathcal{H}_{A}).$ This
follows from the fact that any such function is Lipschitz, $|f(x)-f(y)|\leq
\varkappa |x-y|$, and $\left\Vert \Phi _{2}^{(n)}[\rho ]-\Phi _{2}[\rho
]\right\Vert _{1}\rightarrow 0.$ It follows that (\ref{amp}) holds for all
extreme amplifiers $\Phi _{2}$ in the case of polygonal concave functions $f$
. For arbitrary concave $f$ on $[0,1]$ there is a monotonously nondecreasing
sequence of concave polygonal functions $f_{m}$ converging to $f$ pointwise.
Passing to the limit $m\rightarrow \infty $ gives the inequality (\ref{amp})
for arbitrary extreme amplifier, and hence, (\ref{main1}) holds for
arbitrary Gaussian gauge-covariant or contravariant channels.

(ii) The ``only if'' statement in the cases b), c) are obtained from the
decomposition $\Phi =\Phi _{2}\circ \Phi _{1}$ and the relations (\ref
{decomp})-(\ref{decomb}) by applying argument similar to the case of extreme
amplifier. Notice that the conditions on the channel $\Phi $ imply that in
the decomposition $\Phi =\Phi _{2}\circ \Phi _{1}$ the attenuator $\Phi _{1}$
is defined by the matrix $\mathbf{K}_{1}$ such that $0<\mathbf{K}_{1}\mathbf{
K}_{1}^{\ast }<\mathbf{I}_{A}$ (see Remark \ref{rem}). Applying the argument
involving the relations (\ref{ABBA})-(\ref{AAB}) with strictly concave $f$
to the relations (\ref{decomp})-(\ref{decomc}), we obtain that for any pure
minimizer $P_{\psi }$ of $\mathrm{Tr}f(\Phi [ P_{\psi }])$ the output of the
extremal attenuator $\Phi _{1}[P_{\psi }]$ is necessarily a pure state.
Applying Lemma \ref{2} to the attenuator $\Phi _{1}$ we conclude that $
P_{\psi }$ is necessarily a coherent state.
\end{proof}

\subsection{Explicit formulas and additivity}

\begin{proposition}
\label{dec copy(1)} For any $p>1$ and any Gaussian gauge-covariant or
contravariant channel $\,\Phi $
\begin{eqnarray}
\left\Vert \Phi \right\Vert _{1\rightarrow p} &=&\left( \mathrm{Tr}\Phi [
\rho _{0}]^{p}\right) ^{1/p},  \label{R1} \\
\check{R}_{p}(\Phi ) &=&R_{p}(\Phi [ \rho _{0}]),  \label{R2} \\
\check{H}(\Phi ) &=&H(\Phi [ \rho _{0}]),  \label{R3}
\end{eqnarray}
where $\rho _{0}$ is the vacuum state.

The multiplicativity property (\ref{multi}) holds for any two Gaussian
gauge-covariant (contravariant) channels $\Phi _{1}$ and $\Phi _{2}$, as
well as the additivity of the minimal R\'{e}nyi entropy (\ref{n2p}) and of
the minimal von Neumann entropy (\ref{4}).
\end{proposition}

\begin{proof}
The first statement follows from Theorem \ref{T1} by taking $\ f(x)=-x^{p},$
$\ f(x)=-x\log x.$

If $\Phi _{1}$ and $\Phi _{2}$ are both gauge-covariant (contravariant),
then their tensor product $\Phi _{1}\otimes \Phi _{2}$ shares this property.
The second statement then follows from the expressions (\ref{R1}) - (\ref{R3}
)  and the product property of the vacuum state $\rho_0=\rho_0^{(1)}\otimes
\rho_0^{(2)}$, which follows from the definition.
\end{proof}

From the definitions of gauge-co/contravariant channels (\ref{defprima}), (
\ref{CONTRAV}), it follows that the state $\Phi [\rho _{0}]$ is
gauge-invariant Gaussian with the correlation matrix $\boldsymbol{\mu} +
\mathbf{K}^{\ast } \mathbf{K}/2.$ The spectrum of $\Phi [\rho _{0}]$ is
computed explicitly leading to the expressions \cite{hir}
\begin{equation*}
\left\Vert \Phi \right\Vert _{1\rightarrow p}=\left[ \det \left[ \left(
\boldsymbol{\mu} +\mathbf{K}^{\ast }\mathbf{K}/2+\mathbf{I}_{B}/2\right)
^{p}- \left( \boldsymbol{\mu} +\mathbf{K}^{\ast }\mathbf{K}/2-\mathbf{I}
_{B}/2\right) ^{p}\right] \right] ^{-1/p}
\end{equation*}
and
\begin{equation}
\check{H}(\Phi )=\mathrm{tr}\,g(\boldsymbol{\mu +}\left( \mathbf{K}^{\ast }
\mathbf{K}-\mathbf{I}_{B}\right) /2),  \label{hmin}
\end{equation}
where $g(x)=(x+1)\log (x+1)-x\log x$ and $\mathrm{tr}$ denotes trace of
operators in $\mathbf{Z}.$ In the last case we used the formula for the
entropy of Gaussian state (\ref{gausstate}) \cite{WATER}:
\begin{equation*}
H\left( \rho \right) =\mathrm{tr}\,g(\boldsymbol{\alpha} -\mathbf{I}/2).
\end{equation*}

We now turn to the classical capacity of the channel $\Phi $. In infinite
dimensions, there are two novel features as compared to the situation
described in Sec. \ref{sec:chcap}. First, one has to extend the notion of
ensemble to embrace continual families of states. We call \textit{
generalized ensemble} an arbitrary Borel probability measure $\pi $ on $
\mathfrak{S}(\mathcal{H}_{A})$. The \textit{average state} of the
generalized ensemble $\pi $ is defined as the barycenter of the probability
measure
\begin{equation*}
\bar{\rho}_{\pi }=\int\limits_{\mathfrak{S}(\mathcal{H}_{A})}\rho \,\pi
(d\rho ).
\end{equation*}
The conventional ensembles correspond to finitely supported measures.

Second, one has to consider the input constraints to avoid infinite values
of the capacities. Let $F$ be a positive selfadjoint operator in $\mathcal{H}
_{A}$, which usually represents energy in the system $A$. We consider the
input states with constrained energy: $\mathrm{Tr}\rho F\leq E,$ where $E$
is a fixed positive constant. Since the operator $F$ is usually unbounded,
care should be taken in defining the trace; we put $\mathrm{Tr}\rho
F=\int_{0}^{\infty }\lambda \,dm_{\rho }(\lambda ),$ where $m_{\rho
}(\lambda )=\mathrm{Tr}\rho E(\lambda ),$ and $E(\lambda )$ is the spectral
function of the selfadjoint operator $F.$ Then the constrained $\chi -$
capacity is given by the following generalization of the expression (\ref{1}
):
\begin{equation}
C_{\chi }(\Phi ,F,E)=\sup_{\pi :\mathrm{Tr}\bar{\rho}_{\pi }F\leq E}\chi
(\pi ),  \label{chi}
\end{equation}
where
\begin{equation}
\chi (\pi )=H(\Phi [ \bar{\rho}_{\pi }])-\int\limits_{\mathfrak{S}(\mathcal{H
}_{A})}H(\Phi [ \rho ])\pi (d\rho )  \label{chipi}
\end{equation}
To ensure that this expression is defined correctly, certain additional
conditions upon the channel $\Phi $ and the constraint operator $F$ should
be imposed (see \cite{h}, Sec. 11.5), which however are always fulfilled in
the Gaussian case we consider below.

Denote $F^{(n)}=F\otimes I\dots \otimes I+\dots +I\otimes \dots \otimes
I\otimes F,$ then the \textit{constrained classical capacity} is given by
the expression
\begin{equation}
C(\Phi ,F,E)=\lim_{n\rightarrow \infty }\frac{1}{n}C_{\chi }(\Phi ^{\otimes
n},F^{(n)},nE).  \label{ccc}
\end{equation}

Now let $\Phi $ be a Gaussian gauge-covariant channel, and consider
gauge-invariant oscillator energy operator $F=\sum_{j,k=1}^{s_{A}}\epsilon
_{jk}a_{j}^{\ast }a_{k},$ where $\mathbf{\epsilon }=\left[ \epsilon _{jk}
\right] $ is a Hermitian positive definite matrix, $a_{j}=\frac{1}{\sqrt{2}}(q_{j}+ip_{j})$ -- the annihilation operator
for $j$-th mode. For any state $\rho $
satisfying $\mathrm{Tr}\rho F<\infty ,$ the first moments $\mathrm{Tr}\rho
a_{j}$ and the second moments $\quad \mathrm{Tr}\rho a_{j}^{\ast }a_{k},
\mathrm{Tr}\rho a_{j}a_{k}$ are well defined. For gauge-invariant state $
\mathrm{Tr}\rho a_{j}=0$ and $\mathrm{Tr}\rho a_{j}a_{k}=0.$ For a Gaussian
gauge-invariant state (\ref{gausstate})
\begin{equation*}
\boldsymbol{\alpha} -\mathbf{I}/2=\left[ \mathrm{Tr}\bar{\rho}_{\pi
}a_{j}^{\ast }a_{k}\right] _{j,k=1,\dots ,s},
\end{equation*}
see e.g. \cite{aspekty}.

\begin{proposition}
The constrained classical capacity of the Gaussian gauge-covariant channel $
\Phi $ is
\begin{eqnarray}
C(\Phi ;F,E) &=&C_{\chi }(\Phi ;F,E)  \label{maxnu} \\
&=&\max_{\boldsymbol{\nu }:\,\mathrm{tr}\boldsymbol{\nu \epsilon }\leq E}\,
\mathrm{\ tr\,}g(\mathbf{K}^{\ast }\boldsymbol{\nu}\mathbf{K}+
\boldsymbol{\mu} +\left( \mathbf{K}^{\ast }\mathbf{K}-\mathbf{I}_{B}\right)
/2)- \mathrm{tr}g(\boldsymbol{\mu} +\left( \mathbf{K}^{\ast }\mathbf{K}-
\mathbf{I}_{B}\right) /2).  \notag
\end{eqnarray}
The optimal ensemble $\pi $ which attains the supremum in (\ref{chi})
consists of coherent states $\rho _{\mathbf{z}}=D_{A}(\mathbf{z})\rho
_{0}D_{A}(\mathbf{z})^{\ast },\,\mathbf{z}\in \mathbf{Z}_{A}$ distributed
with gauge-invariant Gaussian probability distribution $Q_{\boldsymbol{\nu }
}(d^{2s}z)$ on $\mathbf{Z}_{A}$ having zero mean and the correlation matrix $
\boldsymbol{\nu }$ which solves the maximization problem in (\ref{maxnu}).
\end{proposition}

\begin{proof}
Consider a Gaussian ensemble $\pi _{\boldsymbol{\nu }}$ consisting of
coherent states $\rho _{\mathbf{z}}=D_{A}(\mathbf{z})\rho _{\mathbf{0}}D_{A}(
\mathbf{z} )^{\ast },\,\mathbf{z}\in \mathbf{Z}_{A},$ with gauge-invariant
Gaussian probability distribution $Q_{\boldsymbol{\nu }}(d^{2s}z)$ on $
\mathbf{Z}_{A}$ having zero mean and some correlation matrix $
\boldsymbol{\nu .}$ It is defined by the classical characteristic function
\begin{equation*}
\int\limits_{\mathbf{Z}_{A}}\exp \left( 2i\mathrm{Im}\mathbf{w}^{\ast }
\mathbf{z}\right) Q_{\boldsymbol{\nu }}(d^{2s}w)=\exp \left( -\mathbf{z}
^{\ast }\boldsymbol{ \nu} \mathbf{z}\right) .
\end{equation*}
By using the covariance property (\ref{gauscov}) of Gaussian channel, we
have
\begin{equation*}
H(\Phi [ \rho _{\mathbf{z}}])=H(\Phi [ D_{A}(\mathbf{z})\rho _{ \mathbf{0}
}D_{A}(\mathbf{z})^{\ast }])=H(\Phi [ \rho _{0}])= \mathrm{tr}g(
\boldsymbol{\mu }+\left( \mathbf{K}^{\ast }\mathbf{K}-\mathbf{I}_{B}\right)
/2),
\end{equation*}
which does not depend on $\mathbf{z,}$ and hence it gives the value of the
integral term in (\ref{chipi}). Integration of the characteristic functions
of coherent states gives
\begin{equation*}
\mathrm{Tr}\bar{\rho}_{\pi _{\boldsymbol{\nu }}}D_{A}(\mathbf{z})=\exp
\left( - \mathbf{z}^{\ast }\left( \boldsymbol{\nu} +\mathbf{I}_{A}/2\right)
\mathbf{z}\right) .
\end{equation*}
Then $\boldsymbol{\nu }=\left[ \mathrm{Tr}\bar{\rho}_{\pi }a_{j}^{\ast }a_{k}
\right] _{j,k=1,\dots ,s_{A}}$ and $\mathrm{Tr}\bar{\rho}_{\pi _{
\boldsymbol{\nu }}}F=\sum_{j,k=1}^{s}\epsilon _{jk}\mathrm{Tr}\bar{\rho}
_{\pi _{\boldsymbol{\nu }}}a_{j}^{\ast }a_{k}=\mathrm{tr}\boldsymbol{\nu
\epsilon .}$ The state $\Phi [ \bar{\rho}_{\pi _{\boldsymbol{\nu }}}]$ is
gauge-invariant Gaussian with the correlation matrix $\mathbf{K}^{\ast
}\left( \boldsymbol{\nu} +\mathbf{I}_{A}/2\right) \mathbf{K}+\boldsymbol{\mu
}$, hence it has the entropy $\mathrm{tr}\,g(\mathbf{K}^{\ast }
\boldsymbol{\nu} \mathbf{K}+\boldsymbol{\mu } +\left( \mathbf{K}^{\ast }
\mathbf{K}-\mathbf{I}_{B}\right) /2).$ Thus for the Gaussian ensemble $\pi _{
\boldsymbol{\nu }}$
\begin{equation}
\chi (\pi _{\boldsymbol{\nu }})=\mathrm{tr}\,g(\mathbf{K}^{\ast }
\boldsymbol{\nu}\mathbf{K}+\boldsymbol{\mu} + \left( \mathbf{K}^{\ast }
\mathbf{K}-\mathbf{I}_{B}\right) /2)-\mathrm{tr}g(\boldsymbol{\mu} +\left(
\mathbf{K}^{\ast }\mathbf{K}-\mathbf{I}_{B}\right) /2).  \label{chinu}
\end{equation}
Summarizing, we need to show
\begin{equation}
C(\Phi ;F,E)=C_{\chi }(\Phi ;F,E)=\sup_{\boldsymbol{\nu }:\,\mathrm{tr}
\boldsymbol{\nu} \boldsymbol{\epsilon }\leq E}\chi (\pi _{\boldsymbol{\nu }
}).  \label{cnu}
\end{equation}

Let us denote by $\mathcal{G}$ the set of Gaussian~ gauge-invariant states
in $\mathcal{H}_{A}$.

\begin{lemma}
\label{L6}
\begin{equation}
\max_{\rho ^{(n)}:\mathrm{Tr}\rho ^{(n)}F^{(n)}\leq nE}H\left( \Phi
^{\otimes n}\left[ \rho ^{(n)}\right] \right) \leq n\max_{\rho :\rho \in
\mathcal{G},\,\mathrm{Tr}\rho F\leq E}H\left( \Phi \left[ \rho \right]
\right) .  \label{maxen}
\end{equation}
\end{lemma}

\begin{proof}
We first prove that
\begin{equation}
\sup_{\rho ^{(n)}:\mathop{\rm Tr}\nolimits\rho ^{(n)}F^{(n)}\leq nE}H(\Phi
^{\otimes n}[\rho ^{(n)}])\leq n\sup_{\rho :\mathop{\rm Tr}\nolimits\rho
F\leq E}H(\Phi [ \rho ]).  \label{hphine}
\end{equation}

Indeed, denoting by $\rho _{j}$ the partial state of $\rho ^{(n)}$ in the $
j- $th tensor factor of ${\mathcal{H}}_{A}^{\otimes n}$ and letting $\bar{
\rho}=\frac{1}{n}\sum_{j=1}^{n}\rho _{j},$ we have
\begin{equation*}
H(\Phi ^{\otimes n}[\rho ^{(n)}])\leq \sum_{j=1}^{n}H(\Phi [ \rho _{j}])\leq
nH(\Phi [ \bar{\rho}]),
\end{equation*}
where in the first inequality we used subadditivity of the quantum entropy,
while in the second -- its concavity. Moreover, $\mathop{\rm Tr}\nolimits
\bar{\rho}F=\frac{1}{n}\mathop{\rm Tr}\nolimits\rho ^{(n)}F^{(n)}\leq E,$
hence (\ref{hphine}) follows.

Using gauge covariance of the channel $\Phi $, we can then reduce
maximization in the right hand side of (\ref{hphine}) to gauge-invariant
states. Indeed, for a given state $\rho$, satisfying the constraint $\mathrm{
Tr}\rho F\leq E$ the averaging
\begin{equation*}
\rho _{av}=\frac{1}{2\pi }\int_{0}^{2\pi }U_{\varphi }\rho U_{\varphi
}^{\ast }d\varphi
\end{equation*}
also satisfies the constraint, while $H(\Phi [ \rho ])\leq H(\Phi [ \rho
_{av}])$ by concavity of the entropy.

Finally, we use the maximum entropy principle which says that among states
with fixed second moments the Gaussian state has maximal entropy (see e.g.
\cite{h}, Lemma 12.25). This proves (\ref{maxen}).
\end{proof}

We have
\begin{equation}
\max_{\rho :\rho \in \mathcal{G},\mathrm{Tr}\rho F\leq E}H\left( \Phi \left[
\rho \right] \right) =\max_{\boldsymbol{\nu }:\mathrm{tr}\boldsymbol{\nu
\epsilon }\leq E}\,\mathrm{tr\,}g(\mathbf{K}^{\ast }\boldsymbol{\nu}\mathbf{K
}+\boldsymbol{\mu} +\left( \mathbf{K}^{\ast }\mathbf{K}-\mathbf{I}
_{B}\right) /2)].  \label{maxg}
\end{equation}
Now let $\boldsymbol{\nu }$ be the solution of the maximization problem in
the righthand side. To prove (\ref{maxnu}) observe that
\begin{eqnarray*}
&&n\chi (\pi _{\boldsymbol{\nu }})\leq nC_{\chi }(\Phi ,F,E)\leq C_{\chi
}(\Phi ^{\otimes n},F^{(n)},nE) \\
&&\qquad \qquad \leq \max_{\rho ^{(n)}:\mathrm{Tr}\rho ^{(n)}F^{(n)}\leq
nE}H\left( \Phi ^{\otimes n}\left[ \rho ^{(n)}\right] \right) -\min_{\rho
^{(n)}}H\left( \Phi ^{\otimes n}\left[ \rho ^{(n)}\right] \right) .\qquad
\qquad
\end{eqnarray*}

By using Lemma \ref{L6} and Proposition \ref{dec copy(1)} we see that this
is less than or equal to
\begin{equation*}
n\left[ \max_{\rho :\rho \in \mathcal{G},\mathrm{Tr}\rho F\leq E}H\left(
\Phi \left[ \rho \right] \right) -H\left( \Phi \left[ \rho _{0}\right]
\right) \right] =n\chi (\pi _{\boldsymbol{\nu }}),
\end{equation*}
where the equality follows from (\ref{maxg}) and (\ref{chinu}).

Thus $C_{\chi }(\Phi ^{\otimes n},F^{(n)},nE)=nC_{\chi }(\Phi ,F,E)$ and
hence the constrained classical capacity (\ref{ccc}) of the Gaussian
gauge-covariant channel is given by the expression (\ref{maxnu}).
\end{proof}

Similar argument applies to Gaussian gauge-contravariant channel (\ref
{CONTRAV}), giving the expression (\ref{maxnu}) with $\boldsymbol{\epsilon }$
replaced by $\bar{\boldsymbol{\epsilon}}.$ Indeed, in this case the state $
\Phi [\bar{\rho}_{\pi _{\boldsymbol{\nu }}}]$ is gauge-invariant Gaussian
with the characteristic function
\begin{eqnarray*}
\mathrm{Tr}\Phi [ \bar{\rho}_{\pi _{\boldsymbol{\nu }}}]D(\mathbf{z}
)&=&\exp \left( -(\overline{\mathbf{Kz}})^{\ast }\left( \boldsymbol{\nu} +
\mathbf{I}_{A}/ 2\right) \overline{\mathbf{Kz}}-\mathbf{z}^{\ast }
\boldsymbol{\mu z}\right) \\
& =&\exp \left( -(\mathbf{Kz})^{\ast }\left( \bar{\boldsymbol{\nu}}+\mathbf{I
}_{A}/2\right) \mathbf{Kz}-\mathbf{z}^{\ast }\boldsymbol{\mu z}\right) ,
\end{eqnarray*}
with the correlation matrix $\mathbf{K}^{\ast }\left(\bar{\boldsymbol{\nu}}+
\mathbf{I}_{A}/2\right) \mathbf{K}+\boldsymbol{\mu} .$ On the other hand, $
\mathrm{tr}\boldsymbol{\nu \epsilon } =\mathrm{tr}\bar{\boldsymbol{\nu}}\bar{
\boldsymbol{\epsilon}},$ so that redefining $ \bar{\boldsymbol{\nu}}$ as $
\boldsymbol{\nu }$, we get the statement.

The maximization in (\ref{maxnu}) is a finite-dimensional optimization
problem which is a quantum analog of ``water-filling'' problem in classical
information theory, see e.g. \cite{COVER,WATER}. It can be solved explicitly
only in some special cases, e.g. when $\mathbf{K},\boldsymbol{\mu},
\boldsymbol{\epsilon}$ commute, and it is a subject of separate study.

\bigskip

\subsection{The case of quantum-classical Gaussian channel}

\label{maj}

Consider affine map which transforms quantum states $\rho \in \mathfrak{S}(
\mathcal{H})$ into probability densities on $\mathbf{Z}$
\begin{equation}
\rho \rightarrow p_{\rho }(\mathbf{z})=\mathrm{Tr}\rho D(\mathbf{z})\rho
_{0}D(\mathbf{z})^{\ast },  \label{qc}
\end{equation}
where $D(\mathbf{z})$ are the displacement operators, $\rho _{0}$ is the
vacuum state with the quantum characteristic function
\begin{equation*}
\phi _{0}(\mathbf{z})\equiv \mathrm{Tr}\rho _{0}D(\mathbf{z})=\exp \left( -
\frac{1}{2}\mathbf{z}^{\ast }\mathbf{z}\right) ,
\end{equation*}
The function $p_{\rho }(\mathbf{z})$ is bounded by 1 and is indeed a
continuous probability density, the normalization follows from the
resolution of the identity
\begin{equation*}
\int_{\mathbf{Z}}D(\mathbf{z})\rho _{0}D(\mathbf{z})^{\ast }\frac{d^{2s}
\mathbf{z}}{\pi ^{s}}=I.
\end{equation*}

\begin{proposition}
\label{T2} Let $f$ be a concave function on $[0,1],$ such that $f(0)=0,$
then for arbitrary state $\rho $
\begin{equation}
\int_{\mathbf{Z}}f(p_{\rho }(\mathbf{z}))\frac{d^{2s}\mathbf{z}}{\pi ^{s}}
\geq \int_{\mathbf{Z}}f(p_{\rho _{\mathbf{w}}}(\mathbf{z}))\frac{d^{2s}
\mathbf{z}}{\pi ^{s}}.  \label{aim}
\end{equation}
\end{proposition}

\begin{proof}
For any $c>0$ consider the\ channel $\Phi _{c}$ defined by the relation
\begin{equation}
\Phi _{c}[\rho ]=\int \frac{d^{2s}\mathbf{z}}{\pi ^{s}c^{2s}}\;\mbox{Tr}
[\rho D(c^{-1}\mathbf{z})\rho _{0}D^{\ast }(c^{-1}\mathbf{z})]\;\rho _{
\mathbf{z}}.  \label{mr}
\end{equation}
The map~(\ref{mr}) is a Gaussian gauge-covariant channel such that
\begin{equation*}
\Phi _{c}^{\ast }[D(\mathbf{z})]=D(c\mathbf{z})\;\exp \left[-\frac{(c^{2}+1)
}{2}|\mathbf{z}|^{2}\right],
\end{equation*}
cf. \cite{ghg}. Therefore by Theorem \ref{T1},
\begin{equation}
\mathrm{Tr}f(\Phi _{c}[\rho ])\geq \mathrm{Tr}f(\Phi _{c}[\rho _{\mathbf{w}
}])  \label{main}
\end{equation}
for all states $\rho $ and any coherent state $\rho _{\mathbf{w}}$. We will
prove the Proposition \ref{T2} by taking the limit $c\rightarrow \infty .$

In the proof we also use a simple generalization of the Berezin-Lieb
inequalities \cite{bl}:
\begin{equation}
\int_{\mathbf{Z}}f(\underline{p}(\mathbf{z}))\frac{d^{2s}\mathbf{z}}{\pi ^{s}
}\leq \mathrm{\ Tr}f(\sigma )\leq \int_{\mathbf{Z}}f(\bar{p}(\mathbf{z}))
\frac{d^{2s}\mathbf{z}}{\pi ^{s}},  \label{blin}
\end{equation}
valid for any quantum state admitting the representation
\begin{equation*}
\sigma =\int_{\mathbf{Z}}\underline{p}(\mathbf{z})\rho _{\mathbf{z}}\frac{
d^{2s}\mathbf{z}}{\pi ^{s}}
\end{equation*}
with a probability density $\underline{p}(\mathbf{z})$. In the right side of
(\ref{blin}) $\bar{p}(\mathbf{z})=\mbox{Tr}\sigma \rho _{\mathbf{z}}.$ In
the inequalities (\ref{blin}) one has to assume that $f$ is defined on $
[0,\infty )$ (in fact, $\underline{p}(\mathbf{z})$ can be unbounded). We
shall assume this for a while.

Taking $\sigma =\Phi _{c}[\rho ],$ from (\ref{mr}) we have
\begin{equation*}
\underline{p}(\mathbf{z})=\frac{1}{c^{2s}}\;\mbox{Tr}\rho D(c^{-1}\mathbf{z}
)\rho _{0}D^{\ast }(c^{-1}\mathbf{z})=\frac{1}{c^{2s}}\;p_{\rho }(c^{-1}
\mathbf{z})\;.
\end{equation*}
while
\begin{equation}
\bar{p}(\mathbf{z})=\mathrm{Tr\,}\rho _{\mathbf{z}}\Phi _{c}[\rho ]=\int_{
\mathbf{Z}}\underline{p}(\mathbf{w})\mathrm{Tr\,}\rho _{\mathbf{z}}\rho _{\mathbf{w}}
\frac{d^{2s}\mathbf{w}}{\pi ^{s}}.  \label{svert}
\end{equation}
We use the well-known formula, see e.g. \cite{KS}, \cite{aspekty},
\begin{equation*}
\mathrm{Tr}\,\rho _{\mathbf{z}}\rho _{\mathbf{w}}=\exp [-|\mathbf{z}-
\mathbf{w}|^2].
\end{equation*}
By introducing the probability density of a normal distribution
\begin{equation*}
q_{c}(\mathbf{z})=\frac{c^{2s}}{\pi ^{s}}\exp \left( -c^{2}|\mathbf{z}
|^2\right)
\end{equation*}
tending to $\delta -$function when $c\rightarrow \infty $ and substituting
this into (\ref{svert}), we have
\begin{eqnarray}
\bar{p}(\mathbf{z}) &=&\int d^{2s}\mathbf{w}\;\underline{p}(\mathbf{w}
)\;q_{1}(\mathbf{z}-\mathbf{w})  \notag \\
&=&\int {d^{2s}\mathbf{w}^{\prime }}\;p_{\rho }(\mathbf{w}^{\prime })\;q_{1}(
\mathbf{z}-c\mathbf{w}^{\prime })  \notag \\
&=&\frac{1}{c^{2s}}p_{\rho }\ast q_{c}(c^{-1}\mathbf{z}).  \label{gg}
\end{eqnarray}

With the change of the integration variable $c^{-1}\mathbf{z}\rightarrow
\mathbf{z}$, the inequalities~(\ref{blin}) become
\begin{equation*}
\int_{\mathbf{Z}}f(c^{-2s}p_{\rho }(\mathbf{z}))\frac{d^{2s}\mathbf{z}}{\pi
^{s}}\leq c^{-2s}\mathrm{Tr}f(\Phi _{c}[\rho ])\leq \int_{\mathbb{C}
^{s}}f(c^{-2s}p_{\rho }\ast q_{c}(\mathbf{z}))\frac{d^{2s}\mathbf{z}}{\pi
^{s}},
\end{equation*}
Substituting $\rho =\rho _{\mathbf{w}},$ we have
\begin{equation*}
\int_{\mathbf{Z}}f(c^{-2s}p_{\rho _{\mathbf{w}}}(\mathbf{z}))\frac{d^{2s}
\mathbf{z}}{\pi ^{s}}\leq c^{-2s}\mathrm{Tr}f(\Phi _{c}[\rho _{\mathbf{w}
}])\leq \int_{\mathbf{Z}}f(c^{-2s}p_{\rho _{\mathbf{w}}}\ast \,q_{c}(\mathbf{
z}))\frac{d^{2s}\mathbf{z}}{\pi ^{s}}.
\end{equation*}
Combining the last two displayed formulas with (\ref{main}) we obtain
\begin{eqnarray}
&&\int_{\mathbf{Z}}g(p_{\rho }(\mathbf{z}))\frac{d^{2s}\mathbf{z}}{\pi ^{s}}
-\int_{\mathbb{\ C}^{s}}g(p_{\rho _{\mathbf{w}}}(\mathbf{z}))\frac{d^{2s}
\mathbf{z}}{\pi ^{s}}  \notag \\
&\geq &\int_{\mathbf{Z}}g(p_{\rho }(\mathbf{z}))\frac{d^{2s}\mathbf{z}}{\pi
^{s}}-\int_{\mathbf{Z}}g(p_{\rho }\ast q_{c}(\mathbf{z}))\frac{d^{2s}\mathbf{
z}}{\pi ^{s}},  \label{ine}
\end{eqnarray}
where we denoted $g(x)=f(c^{-2s}x),$ which is again a concave function.
Moreover, arbitrary concave polygonal function $g$ on $[0,1],$ satisfying $
g(0)=0,$ can be obtained in this way by defining
\begin{equation*}
f(x)=\left\{
\begin{array}{l}
g(c^{2s}x),\quad x\in [ 0,c^{-2s}] \\
g(1)+g^{\prime }(1)(x-c^{-2s}),\quad x\in [ c^{-2s},\infty )
\end{array}
\right. ,
\end{equation*}
hence (\ref{ine}) holds for any such function. Then the right hand side of
the inequality (\ref{ine}) tends to zero as $c\rightarrow \infty .$ Indeed,
for polygonal function $\left\vert g(x)-g(y)\right\vert \leq \varkappa
\left\vert x-y\right\vert ,$ and the asserted convergence follows from the
convergence $p_{\rho }\ast q_{c}\longrightarrow p_{\rho }$ in $L_{1}:$ if $p(
\mathbf{z})$ is a bounded continuous probability density, then
\begin{equation*}
\lim_{c\rightarrow \infty }\int_{\mathbf{Z}}\left\vert p\ast q_{c}(\mathbf{z}
)-p(\mathbf{z})\right\vert d^{2s}\mathbf{z}=0.
\end{equation*}
Thus we obtain (\ref{aim}) for the concave polygonal functions $f.$ But for
arbitrary continuous concave $f$ on $[0,1]$ there is a monotonously
nondecreasing sequence of concave polygonal functions $f_{n}$ converging to $
f$ . Applying Beppo-Levy's theorem, we obtain the statement.
\end{proof}

\section{Appendix}

Consider a gauge-covariant channel $\Phi $ such that the matrices $\mathbf{K}
^{\ast }\mathbf{K}$ and $\boldsymbol{\mu }$ commute (in particular, this
condition is satisfied by extreme amplifiers and attenuators). These
channels are diagonalizable in the following sense. We have
\begin{equation*}
\mathbf{K}=\mathbf{V}_{A}\mathbf{K}_{d}\mathbf{V}_{B},\quad \boldsymbol{\mu} =\mathbf{V}
_{B}^{\ast }\boldsymbol{\mu }_{d}\mathbf{V}_{B}\,,
\end{equation*}
where $\mathbf{V}_{A},\mathbf{V}_{B}$ are unitaries and $\mathbf{K}_{d},\,
\boldsymbol{\mu }_{d}$ are diagonal (rectangular) matrices with nonnegative values on the
diagonal. Then $\mathbf{K}^{\ast }\mathbf{K}=\mathbf{V}_{B}^{\ast }\mathbf{K}
_{d}^{2} \mathbf{V}_{B},$ and
\begin{equation}
\Phi [ \rho ]=U_{B}\Phi _{d}[U_{A}\rho U_{A}^{\ast }]U_{B}^{\ast },
\label{unieq}
\end{equation}
where $U_{A}$, $U_{B}$ are canonical unitary (\textquotedblleft
metaplectic\textquotedblright\ \cite{simon}) transformations acting on $
\mathcal{H}_{A}$, $\mathcal{H}_{B}$ such that
\begin{equation*}
U_{B}^{\ast }D_{B}(\mathbf{z})U_{B}=D_{B}(\mathbf{V}_{B}\mathbf{z}),\qquad \qquad
U_{A}^{\ast }D_{A}(\mathbf{z})U_{A}=D_{A}(\mathbf{V}_{A}\mathbf{z}),
\end{equation*}
To describe the action of \textquotedblleft\ diagonal\textquotedblright\
channel $\Phi _{d}$ in more detail, we have to consider separately the cases
$s_{A}=s_{B},\,s_{A}\leq s_{B}$ and $s_{A}>s_{B}.$

In the case $s_{A}=s_{B}$ we have
\begin{equation*}
\mathbf{K}_{d}=\mathrm{diag}\left[ k_{j}\right] _{j=1,\dots ,s_{B}};\quad
\boldsymbol{\mu }_{d}=\mathrm{diag}\left[ \mu _{j}\right] _{j=1,\dots
,s_{B}}.
\end{equation*}
Then $\Phi _{d}=\otimes _{j=1}^{s_{B}}\Phi _{j},$ where, in self-explanatory
notations,
\begin{equation}
\Phi _{j}^{\ast }[D_{j}(z_{j})]=D_{j}(k_{j}z_{j})\exp \left( -\mu
_{j}\left\vert z_{j}\right\vert ^{2}\right) .  \label{onemod}
\end{equation}

In the case $s_{A}<s_{B}$
\begin{equation*}
\mathbf{K}_{d}=\left[
\begin{array}{c}
\mathrm{diag}\left[ k_{j}\right] _{j=1,\dots ,s_{A}} \\
\mathbf{0}
\end{array}
\right]
\end{equation*}
where $\mathbf{0}$ denotes block of zeroes of the size $\left(
s_{B}-s_{A}\right) \times s_{A}$. Then
\begin{equation*}
\Phi _{d}[\rho ]=\otimes _{j=1}^{s_{A}}\Phi _{j}[\rho ]\otimes \rho
_{0}^{[s_{A}+1,\dots ,s_{B}]},
\end{equation*}
where for $j=1,\dots ,s_{A}$ the one-mode channels $\Phi _{j}$ are given by (
\ref{onemod}), and $\rho _{0}^{[s_{A}+1,\dots ,s_{B}]}$ is the vacuum state
of the modes $s_{A}+1,\dots ,s_{B}.$

In the case $s_{A}>s_{B}$
\begin{equation*}
\mathbf{K}_{d}=\left[
\begin{array}{cc}
\mathrm{diag}\left[ k_{j}\right] _{j=1,\dots ,s_{B}} & \mathbf{0}
\end{array}
\right]
\end{equation*}
where $\mathbf{0}$ denotes block of zeroes of the size $s_{B}\times \left(
s_{A}-s_{B}\right) $, and
\begin{equation*}
\Phi _{d}[\rho ]=\left( \otimes _{j=1}^{s_{A}}\Phi _{j}\right) [\mathrm{Tr}
_{s_{B}+1,\dots ,s_{A}}\rho ],
\end{equation*}
where $\mathrm{Tr}_{s_{B}+1,\dots ,s_{A}}\mathrm{\ }$denotes partial trace
over the last $s_{A}-s_{B}$ modes of the operator $\rho .$

There is a similar reduction to the diagonal form for gauge-contravariant
channels.

\section{Acknowledgments}

The author is grateful to M. E. Shirokov (Steklov Mathematical Institute) for comments and discussions. Thanks are due to David Ding (Stanford University)
for pointing out some typos. The work was supported by the grant of Russian Scientific Foundation  (project No 14-21-00162).

\end{document}